\documentclass[12pt,a4paper,reqno]{amsart}
\usepackage{amsmath, amssymb, amsfonts, amsthm}
\usepackage{hyperref}
\usepackage[english]{babel}
\usepackage{amsmath}
\usepackage{amssymb}
\usepackage{url}
\usepackage{amsthm}
\usepackage{amsxtra}
\usepackage{mathrsfs}
\usepackage{fancyhdr}
\usepackage{enumerate}
\usepackage{graphicx}
\usepackage{hyperref}
\usepackage{color}
\usepackage{tikz}
\usepackage{bm}
\usepackage{geometry}
\newtheorem{theorem}{Theorem}[section]

\newtheorem{lemma}{Lemma}[section]
\newtheorem{proposition}{Proposition}[section]
\newtheorem{corollary}{Corollary}[section]
\newtheorem{definition}{Definition}[section]

\theoremstyle{definition}
\newtheorem{remark}[lemma]{Remark}
\newtheorem{assumption}{Assumption}

\numberwithin{equation}{section}

\usepackage[deletedmarkup=xout,
]{changes}

\definechangesauthor[name={Domenico}, color={blue}]{Domenico}

\definechangesauthor[name={Raffaele}, color={green}]{Raffaele}

\newcommand{\beq}{\begin{equation}}
\newcommand{\eeq}{\end{equation}}
\newcommand{\be}{\begin{equation*}}
\newcommand{\ee}{\end{equation*}}
\newcommand{\n}{\noindent}

\newcommand{\RE}{\mathbb R}
\newcommand{\erre}{\mathbb R}

\newcommand{\BB}{\mathscr B}
\newcommand{\DD}{\mathscr D}
\newcommand{\OO}{\mathcal{O}}

\newcommand{\LL}{\mathcal{L}}

\newcommand{\sgn}{\operatorname{sgn}\,}

\newcommand{\lf}{\left}
\newcommand{\ri}{\right}
\newcommand{\ve}{\varepsilon}
\newcommand{\al}{\alpha}

\newcommand{\la}{\lambda}
\newcommand{\de}{\delta}

\newcommand{\lan}{\langle}
\newcommand{\ran}{\rangle}
\renewcommand{\Im}{\operatorname{Im}\,}
\renewcommand{\Re}{\operatorname{Re}\,}

\renewcommand{\leq}{\leqslant}
\renewcommand{\geq}{\geqslant}

\newcommand{\x}{\underline{x}}

\newcommand{\f}{\frac}
\newcommand{\EE}{\mathcal E}
\newcommand{\HH}{\mathcal H}

\newcommand{\WW}{\mathcal W}

\newcommand{\xv}{\mathbf{x}}
\newcommand{\kv}{\mathbf{k}}
\newcommand{\pv}{\mathbf{p}}
\newcommand{\kvp}{\mathbf{k}^{\prime}}
\newcommand{\xvp}{\mathbf{x}^{\prime}}

\newcommand{\yv}{\mathbf{y}}

\newcommand{\bdm}{\begin{displaymath}}
\newcommand{\edm}{\end{displaymath}}
\newcommand{\bdn}{\begin{eqnarray}}
\newcommand{\edn}{\end{eqnarray}}
\newcommand{\bay}{\begin{array}{c}}
\newcommand{\eay}{\end{array}}
\newcommand{\ben}{\begin{enumerate}}
\newcommand{\een}{\end{enumerate}}
\newcommand{\beqn}{\begin{eqnarray}}
\newcommand{\eeqn}{\end{eqnarray}}
\newcommand{\bml}[1]{\begin{multline} #1 \end{multline}}
\newcommand{\bmln}[1]{\begin{multline*} #1 \end{multline*}}

\newcommand{\R}{\mathbb{R}}
\newcommand{\N}{\mathbb{N}}

\newcommand{\diff}{\mathrm{d}}
\newcommand{\dom}{\mathscr{D}}
\newcommand{\eps}{\varepsilon}

\newcommand{\amp}{A_V(\kv , \kvp)}
\newcommand{\ampe}{A_{\mathrm{eff}}(\kv , \kvp; \la)}
\newcommand{\ampeo}{A_{\mathrm{eff}}(0, 0; \la)}
\newcommand{\eigenv}{\phi_{V,\kv}}
\newcommand{\aeff}{a_{\mathrm{eff}}(\lambda)}

\newcommand{\tx}{\textstyle}
\newcommand{\nt}{\noindent}
\newcommand{\bra}[1]{\lf\langle #1\ri|}
\newcommand{\ket}[1]{\lf|#1 \ri\rangle}
\newcommand{\braket}[2]{\lf\langle #1|#2 \ri\rangle}
\newcommand{\braketr}[2]{\lf\langle #1\lf|#2\ri. \ri\rangle}
\newcommand{\braketl}[2]{\lf.\lf\langle #1\ri|#2 \ri\rangle}
\newcommand{\mean}[3]{\bra{#1}#2\ket{#3}}
\newcommand{\meanlrlr}[3]{\lf\langle #1\lf|#2\ri|#3\ri\rangle}
\newcommand{\meanlr}[3]{\lf\langle #1\Big|#2\Big|#3\ri\rangle}

\begin{document}

\title{A Quantum Model of Feshbach Resonances}

\author{R. Carlone}
\address{Dipartimento di Matematica e Applicazioni ``R. Caccioppoli'', Universit\'a di Napoli ``Federico II'', Via Cinthia, Monte S. Angelo, 80126 Napoli, Italy}

\author{M. Correggi}
\address{Dipartimento di Matematica, ``Sapienza'' Universit\`a di Roma, P.le A. Moro 5, 00185 Roma, Italy}
\email{michele.correggi@gmail.com}
\urladdr{http://www1.mat.uniroma1.it/people/correggi/}

\author{D. Finco}
\address{Facolt\`a di Ingegneria, Universit\`a Telematica Internazionale Uninettuno, Corso V. Emanuele II 39,  00186 Roma, Italy}
\author{A. Teta}
\address{Dipartimento di Matematica, ``Sapienza'' Universit\`a di Roma, P.le A. Moro 5, 00185 Roma, Italy}

\date{\today}

\subjclass[2010]{}

\begin{abstract}
We consider a quantum model of two-channel scattering to describe the mechanism of a Feshbach resonance. We perform a rigorous analysis in order to count and localize the energy resonances in the perturbative regime, i.e., for small inter-channel coupling, and in the non-perturbative one. We provide an expansion of the effective scattering length near the resonances, via a detailed study of an effective Lippmann-Schwinger equation with energy-dependent potential.
\end{abstract}    


\maketitle

\section{Introduction}
\label{sec: intro}

In the last twenty years, a central topic in Physics and Mathematical Physics has been the study of ultra-cold quantum gases.
The primary motivation is the extraordinary degree of control that such systems offer and that allows to investigate the fundamental behavior of quantum matter
in various regimes.

In an ultra-cold quantum gas, the particle density is such that only the two-body scattering process is relevant at low energy, i.e., in the $s$-wave, and the scattering length $a$ is the only physical parameter characterizing the inter-particle interaction. This single parameter can, in turn, be controlled concretely in the experiments via the so-called  Feshbach resonance mechanism \cite{DS,fis10, fis11,fis3,fis4,S,fis9,Ti}, which allows to set the scattering length to any value, both positive and negative. 

The role of scattering length tunability in Bose-Einstein condensates is apparent, since the size of the  scattering length is crucial in order to reach condensation. Furthermore, in the Thomas-Fermi regime, the essential physical features of the condensate, as, e.g., the radius of its support or the highest peak of the particle density, strongly depends on the scattering length and in fact scales as suitable powers of it. Even more, in a two-component Fermi gas, one can tune $ a $ from positive to negative values, so passing from a regime of molecular Bose-Einstein condensation to the formation of a BCS superfluid of Cooper-like pairs. One can even single out the transition region, where the scattering length diverges and the gas is expected to show universal features: in this case, the system goes under the name of unitary gas \cite{fis7,CFT,fis6} and we refer to \cite{CDFMT1,CDFMT2} and references therein for further information.

The relevance of this control mechanism is however not limited to the physics of cold atoms but in fact emerges in many other applications, as, e.g., scattering by nuclei or nuclear physics, or for the realization of the Efimov effect \cite{efimov,fis8} (see also \cite{BT} and references therein for rigorous results), i.e., when the two-body interaction becomes resonant, and Efimov trimers appear.

The Feshbach resonance mechanism for cold alkali atoms can be easily modelled in terms of a {\it multi-channel scattering process} \cite[Chpts. 3 \& 5]{PS}. Indeed, the internal atomic structure is then given by the usual energy levels which are split because of the hyperfine coupling and thus it is basically determined by the properties of the valence electrons and specifically their spin. As a consequence, the interaction between two atoms induces transitions between the hyperfine energy levels and thus changes the spin of such electrons. It is therefore evident that such a mechanism is conveniently described by a multi-channel Hamiltonian. In this context, for a given value of the spin, one speaks of an {\it open channel}, when for that value of the spin a scattering process is allowed, whereas  the channel is called {\it closed}, if no scattering is possible and the system is described by a bound state (see Fig. \ref{graph}). A Feshbach resonance occurs when the low energy threshold (ionization threshold) of the scattering process in the open channel is close to the energy of the closed channel. Under such resonance condition and in presence of an inter-channel interaction, the scattering process in the open channels is strongly influenced by the presence of the closed channel, and the corresponding scattering length is modified accordingly, even for weak coupling between the channels. 

In extreme synthesis, the Feshbach resonance mechanism described above can be physically described as follows:  let us consider the stationary Schr\"odinger equation
\beq \label{fesh1}
{\mathscr H} \Psi = E \Psi
\eeq
and let us introduce two orthogonal projectors $P$ and $Q$ such that $P+Q=Id$. In the multi-channel picture $ P $ is the projector onto the open channel, while $ Q $ projects on the closed one. Equation \eqref{fesh1} is equivalent to the system
\begin{align}
&( {\mathscr H}_{PP} - E ) P\Psi + {\mathscr H}_{PQ} Q\Psi =0 \label{fesh2}\\
& ( {\mathscr H}_{QQ} - E ) Q\Psi + {\mathscr H}_{QP} P\Psi =0  \label{fesh3}
\end{align}
where ${\mathscr H}_{PP}= P {\mathscr H}P$, ${\mathscr H}_{PQ}= P {\mathscr H}Q$ and  ${\mathscr H}_{QQ}= Q{\mathscr H}Q$. The first channel is  open when $E$ lies in the continuous spectrum of ${\mathscr H}_{PP}$ and it is associated with a scattering process, i.e., $ E $ is not an eigenvalue of ${\mathscr H}_{PP} $. The second channel is  closed when $E$ lies below the continuous spectrum threshold of $\mathscr{H}_{QQ}$.
In such a situation, $ {\mathscr H}_{QQ} - E $ is invertible and equation \eqref{fesh3} can be solved, so that, after substitution in \eqref{fesh2}, \eqref{fesh1} proves to be equivalent to
\beq \label{fesh4}
{\mathscr H}_{\text{eff}} P\Psi = E P\Psi, \qquad
{\mathscr H}_{\text{eff}} = {\mathscr H}_{PP}- {\mathscr H}_{PQ} \lf( {\mathscr H}_{QQ} - E \ri)^{-1} {\mathscr H}_{QP},
\eeq
i.e., we have reduced the eigenvalue problem \eqref{fesh1} to an effective equation in the open channel.
Under the hypothesis that the interaction between channels is weak, \eqref{fesh3} can be approximated by $  ( {\mathscr H}_{QQ} - E ) Q\Psi =0$. Suppose now that $E_{\text{b}}$ and $\Phi_{\text{b}}$ satisfy
\beq \label{fesh5}
( {\mathscr H}_{QQ} - E_{\text{b}} ) \Phi_{\text{b}}=0,
\eeq
then, for $E$ near $E_{\text{b}}$, we can approximate ${\mathscr H}_{\text{eff}}$ by
\beq \label{fesh6}
{\mathscr H}_{\text{eff}}' = {\mathscr H}_{PP} - \dfrac{\ket{ {\mathscr H}_{PQ}  \Phi_{\text{b}}}  \bra{ {\mathscr H}_{QP}\Phi_{\text{b}}}}{E- E_{\text{b}}}.
\eeq 
This substitution is usually called single pole approximation. Let us we denote by $\Phi^+ $ the outgoing solution of $ ( {\mathscr H}_{PP} - E ) \Phi=0 $. Since ${\mathscr H}_{\text{eff}}'$ is a rank one perturbation of ${\mathscr H}_{PP}$, it is straightforward to see that the solution of \eqref{fesh4} can be approximated by
\beq \label{fesh7}
P\Psi = \Phi^+  -\dfrac{ ( {\mathscr H}_{PP} - (E+i0) )^{-1} \ket{{\mathscr H}_{PQ} \Phi_{\text{b}}} \bra{ {\mathscr H}_{QP}\Phi_{\text{b}}}}{E- E_{\text{b}}-  \bra{{\mathscr H}_{QP} \Phi_{\text{b}} } ({\mathscr H}_{PP} - (E+i0) )^{-1} \ket{{\mathscr H}_{PQ} \Phi_{\text{b}}}},
\eeq
where $ ({\mathscr H}_{PP} - (E+i0) )^{-1} $ stands for the boundary value of the resolvent.
Then, we see that the resonant scattering amplitude contains a {\it Breit-Wigner pole} $(E-E_{\text{res}}+i\Gamma/2)^{-1}$ with
\begin{align*}
E_{\text{res}} &= E_{\text{b}} + \Re \meanlrlr{{\mathscr H}_{QP}\Phi_{\text{b}}}{  \lf( {\mathscr H}_{PP} - (E+i0) \ri)^{-1}}{{\mathscr H}_{PQ} \Phi_{\text{b}}}, \\
\Gamma &= -2\Im  \meanlrlr{{\mathscr H}_{QP}\Phi_{\text{b}}}{ \lf( {\mathscr H}_{PP} - (E+i0) \ri)^{-1}}{{\mathscr H}_{PQ} \Phi_{\text{b}}}.
\end{align*}
Notice that $E_{\text{res}}$ is shifted from $ E_{\text{b}}$ due to channel interaction.

It is worth mentioning  that a  typical way to produce the channel splitting for cold atoms is obtained through an external magnetic field (magnetically induced Feshbach resonances).   In this case, the hyperfine energy levels associated to different configurations of the spin degrees of freedom split 
due to Zeeman effect and the splitting is proportional to the external magnetic field. If we normalize the zero energy level as the continuum threshold
of the open channel, the net effect of the magnetic field is to add a constant term  $\Delta E$ proportional to $B$ in $ {\mathscr H}_{QQ}$. 
As a consequence, we can set $B$ in such a way that $ E_{\text{b}}$ is zero, and the resonance influences the zero energy behavior of the open channel. The effective coupling between the channels is provided by the typical Van der Waals-like potentials. Hence, the scattering length is given by the formula
\beq \label{fesh8}
a(B)= a_{\infty}+ \frac{\Delta}{B- B_{\text{res}}},
\eeq
where $B_{\text{res}}$ is the resonant value of $B$, $\Delta$ is the resonance width and $a_{\infty}$ is the asymptotic value of $a$ for large $B$.

\medskip

Matrix Hamiltonians like $ \mathscr{H} $ above have already been considered in the mathematical literature \cite{djn1,djn2} in the abstract setting but, typically, in a simplified form, i.e., with
\[
{\mathscr H}_{QQ} = E_{\text{b}} \ket{\Phi_{\text{b}}}\bra{\Phi_{\text{b}}},
\]
which is equivalent to a single pole approximation of the closed channel. In this way, the total Hamiltonian $ \mathscr{H} $ may have an eigenvalue embedded in the continuous spectrum or at the threshold. The survival probability of this eigenvalue is  studied in the perturbative regime where the inter channel interaction is weak, and it is shown to depend on the properties of the function
\[
F(E)= E- E_{\text{b}}- \bra{{\mathscr H}_{QP} \Phi_{\text{b}} } ({\mathscr H}_{PP} - (E+i0) )^{-1} \ket{{\mathscr H}_{PQ} \Phi_{\text{b}}}.
\]
Other similar results on the Fermi Golden Rule are discussed in \cite{cjn, jn}, while the behavior of  eigenvalues at threshold is studied, e.g., in \cite{m1,m2} for small perturbations.

\medskip

In this paper, we consider a two-channel model, and we perform a rigorous analysis providing some information concerning both the number and the localization of Feshbach resonances.
It is indeed clear that the heuristic argument sketched above is not satisfying from the mathematical point of view: the Born expansion, implicitly assumed in \eqref{fesh4} and \eqref{fesh6}, breaks down near the resonance.  Moreover, we prove an expansion of the scattering length near a resonance similar to \eqref{fesh8}. Notice that no assumption of smallness for the inter-channel interaction is assumed in our model. If however, such an assumption does hold, we strengthen the results about the counting and the localization of Feshbach resonances.

\medskip

The paper is organized as follows. In Section \ref{sec:main} we introduce the model  and we state our main results. A simplified solvable model where the inter-channel interaction is given by a   separable potential is presented and worked out in Section \ref{sec:separable}. In Section \ref{sec:notation} we collect some notation and the necessary technical results used in the paper.
The generalized eigenfunctions of the effective problem in the open channel at low energy are discussed in Section \ref{sec:eigenfunctions}.  In Section \ref{sec:Feshbach} we finally prove the characterization of the number and the localization of Feshbach resonances.

\begin{center}
\begin{figure}\label{graph}
\includegraphics[width=8cm]{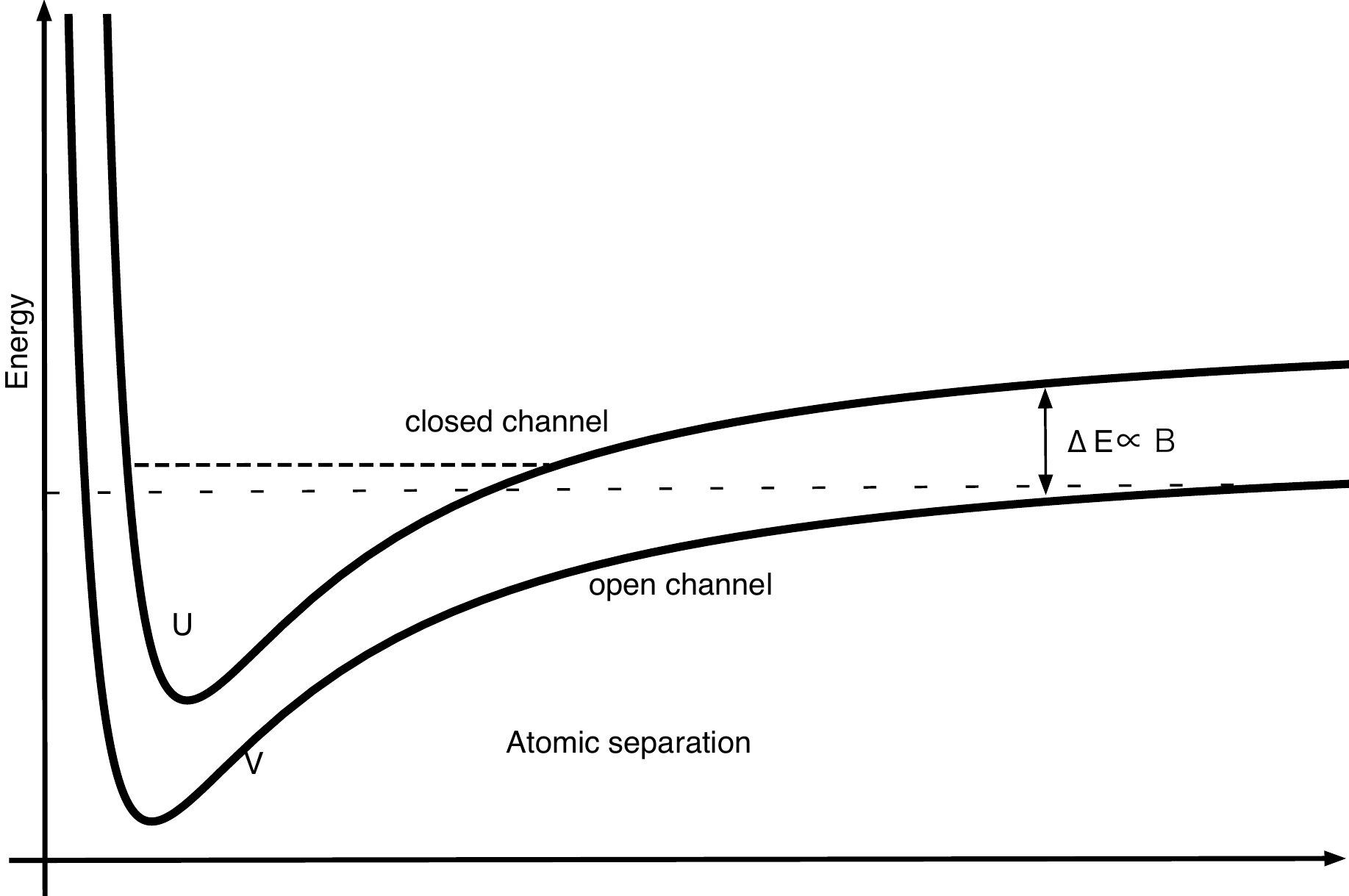} 
\caption{The potentials in open and closed channels are depicted as thick lines. The dashed lines represent the energy of the scattering process and the bound state energy in the closed channel, respectively.}
\end{figure}
\end{center}


\section{Main Results}
\label{sec:main}

In order to formulate our results, we first fix some notational conventions used in the paper. 

\n
We use when possible  capital greek or calligraphic letters to denote either vectors or matrix operators, while lower case or regular capital letters will stand for vector in $ L^2(\RE^3) $ or operator on the same space. Bold face letters denote vectors in $ \R^3 $, e.g., $ \xv $, while scalars are denoted by regular letters, e.g., $ E $. When there is no ambiguity we also use the notation $ x : = |\xv| $ for the modulus of a vector $ \xv $.
Since we always deal with functions on $ \R^3 $, we often omit the base space in Banach space notation, i.e., $ L^p : = L^p(\R^3) $, where there is no possible confusion and denote $ \lf\| \: \cdot \: \ri\|_{L^p} =: \lf\| \: \cdot \: \ri\|_p $.
Finally, by $C$  we denote a generic positive constant which can possibly vary from line to line. 

\medskip

We now make the mathematical setting more precise. We consider a multi-channel scattering of a particle in dimension three: we assume that there are an {\it open channel}, where the scattering is energetically possible, and a {\it closed} one where any scattering process is forbidden because of an energy constraint. The two channels interact via a coupling term in the Hamiltonian. Hence, we can describe the system by the following matrix Hamiltonian 
acting on the Hilbert space  ${\mathscr H} = L^2 (\RE^3) \oplus  L^2 (\RE^3) $:
\beq
\HH = 
\left(
\begin{array}{cc}
H_V & W \\
 W & H_U +\la
\end{array}
\right) = \HH_0 + \WW, 
\eeq
\beq
\HH_0 = 
\left(
\begin{array}{cc}
H_V & 0 \\
 0 & H_U  +\la
\end{array}
\right),
\qquad
\WW = 
\left(
\begin{array}{cc}
0 & W \\
 W & 0
\end{array}
\right), 
\eeq
where 
\beq
	\label{eq:1bh}
	H_V=-\Delta +V,	\qquad		 H_U=-\Delta +U
\eeq
are the open and closed channel Schr\"{o}dinger operators, respectively, and $W$ is the inter-channel interaction, i.e., in the physical picture the hyperfine coupling between different Zeeman levels. The  role of the control parameter is played by $\la > 0$, which can be thought as proportional to the magnetic field.

In order to study the modified scattering in the open channel, we aim at deriving an equation analogous to the usual {\it Lippmann-Schwinger (LS) equation} for the one or two-particle scattering problem. We recall here for the convenience of the reader the basics of time-independent scattering theory of quantum particle by a smooth and short range potential $ V(\xv) $ (see, e.g., \cite{RS3}). In this simple case the spectrum of the Hamiltonian $ H_V $ is known  to be absolutely continuous and to coincide with the positive real line, with the possible exception of some isolated negative eigenvalues. Then, let 
\beq
 \phi_{0,\kv}(\xv):= e^{i \kv \cdot \xv} 
 \eeq
 be a plane wave with momentum 
$ \kv \in \R^3 $, which is a generalized eigenfunction of the unperturbed Hamiltonian $  - \Delta $, then the LS equation is
\beq
	\label{eq:ls}
	\phi_{V,\kv}(\xv) = \phi_{0,\kv}(\xv) - \lf( R_0(k^2)  V  \phi_{V,\kv} \ri)(\xv),
\eeq
where we  denote by $ R_V(E)$  the resolvent of the operator $ H_V $ (which reduces to the resolvent of $- \Delta$ for $V=0$). For $ E$ in the resolvent set of $H_V$  
\beq
	\label{eq:resolvent}
	R_{V}(E) : = \lf( H_V - E \ri)^{-1}
\eeq
is a bounded operator in $L^2$ while, if $ E \geq 0 $ falls into the continuous spectrum of $ H_V $, then $ R_V(E) $ stands for the boundary value of the resolvent from the upper half-plane, 
\beq
	\label{eq:resolvent-bv}
	R_{V}(E) : = \lim_{\eps \to 0^+} \lf( H_V - E - i \eps \ri)^{-1}.
\eeq
The solutions of the LS equation are generalized eigenfunctions of $ H_V $ (a.k.a. distorted plane waves) with energy $ E = k^2 \geq 0 $ lying inside the continuous spectrum. 

Under suitable assumptions on the potential $ V $, one can show (see, e.g.,  \cite{A}) that equation \eqref{eq:ls} has a  unique continuous solution, which asymptotically behaves like
\beq
	\label{eq:ls-solution}
	\phi_{V,\kv}(\xv) \underset{x \to +\infty}{\simeq} e^{i \kv \cdot \xv} + \amp \frac{e^{i k x}}{x},
\eeq
i.e., the superposition of a plane wave and a spherical one. Here we denote by 
\beq
	\label{eq:k-transferred}
	\kvp : = k \hat{\xv},
\eeq
the outgoing momentum. The quantity $ \amp $ is named {\it scattering amplitude} and is in fact sufficient to completely describe the elastic scattering process. At low energy one can actually  characterize all the effects of the scattering by a single scalar quantity, the {\it scattering length} associated to $ V $, which is defined as
\beq
	\label{eq:scattering-length}
	a_V : = \lim_{k \to 0} \amp.
\eeq

The starting point of our investigation is the eigenvalue equation for the matrix Hamiltonian $ \HH $: 
\beq\label{eigenequa}
\HH \Psi_\kv = k^2 \, \Psi_\kv.
\eeq
We are going to cast such an equation in a form akin to the LS equation \eqref{eq:ls} and prove that the corresponding solution has the form \eqref{eq:ls-solution}. This will allow us to introduce an {\it effective scattering length} in the open channel, which takes into account the effect of the interaction with the closed one. We will then characterize further the salient features of such an effective parameter. 

Writing $\Psi_\kv=(\varphi_\kv, \xi_\kv)$,  equation \eqref{eigenequa} is equivalent to the coupled system
\beq \label{geneigen}
\begin{cases}
(-\Delta + V)\varphi_\kv + W \xi_\kv =  k^2 \varphi_\kv, \\    
( -\Delta + U +\la )\xi_\kv+ W \varphi_\kv = k^2 \xi_\kv. 
\end{cases} 
\eeq
Since we are interested in the low energy behavior, we can restrict to the energies  
\beq
	\label{eq:condition-k}
	0<k^2<\la.
\eeq
For $k^2 - \lambda$ in the resolvent set of $H_U$, $ -\Delta + U +\la -k^2$ has a bounded inverse $R_U(k^2-\la)$  in $L^2(\erre^3)$ and the system \eqref{geneigen} is equivalent to the coupled integral equations
\beq \label{geneigen2}
\begin{cases}
 \varphi_\kv + R_V (k^2) \, W\, \xi_\kv = \eigenv, \\
 \xi_\kv + R_U(k^2-\la)\, W\, \varphi_\kv= 0,
\end{cases}
\eeq
where we have denoted by $ \eigenv $ the generalized eigenfunctions of $H_V$, i.e., the solution of the LS equation \eqref{eq:ls}. The resolvent $ R_V(k^2) $ is defined as a boundary value as in \eqref{eq:resolvent-bv}. From \eqref{geneigen2} one sees that the problem is reduced to find the solution  $\varphi_\kv$ of the {\it effective LS equation}
\beq \label{geneigen3}
 \varphi_\kv - R_V (k^2) \, W\, R_U(k^2-\la)\, W\, \varphi_\kv = \eigenv .
\eeq

Before stating the first result, let us specify the technical assumptions we make on the potentials $ U $, $ V $ and $ W $. Our analysis is partly based on the classical work of Ikebe \cite{I} and therefore we recall first the definition of the Ikebe class $ I_n(\R^3) $, $ n \in \N $.

\begin{definition}[Ikebe class $ I_n $]
	\label{def:ikebe}
	\mbox{}	\\
We say that a measurable function $V$ belongs to the Ikebe class $I_n(\R^3)$, $ n \in \N $, if $V\in L^2(\R^3)$, $ V $ is locally H\"older continuous except for a finite number of points and there exists $R_0>0$ and $\de\in (0,1)$ such that 
	\beq
		\label{eq:ikebe-decay}
		|V(\xv)| \leq \dfrac{C}{x^{n+\de}},	\qquad		\mbox{for } x\geq R_0.
	\eeq
\end{definition}

\begin{assumption} \label{pot}
We assume that 
	\begin{itemize}
		\item[a)] $U\in I_2(\R^3)$;
		\item[b)] $V \in I_4(\R^3) \cap L^3(\R^3) $;
		\item[c)] $W \in I_3(\R^3) \cap L^3(\R^3) $. 
	\end{itemize}
\end{assumption}

Our assumptions on the potentials are certainly not optimal and other hypothesis are possible (see, e.g., \cite{RS3,SQF}). However, we stick the choice above for concreteness. Note that, under Assumption \ref{pot}, $ V $ and $ U $ are {\it Agmon potentials} according to \cite[Sect. XIII.8]{RS4} and therefore are relatively compact perturbations of $-\Delta$. In particular, the one-body Hamiltonians $H_U $ and  $H_V $ are self-adjoint on $H^2(\R^3) $. Furthermore, 
\begin{itemize}
	\item $H_U $ and $ H_V $ have finitely many negative eigenvalues by, e.g., the Birman-Schwinger bound \cite[Thm. XIII.10]{RS4} (see also \cite[Thm. 4]{I});
	\item $ \sigma_{\mathrm{ess}}(H_U) = \sigma_{\mathrm{ess}}(H_V) = [0,+\infty) $ by, e.g.,  Weyl's theorem \cite[Thm. XIII.15]{RS4};
	\item $H_U $ and $ H_V $ have no positive eigenvalues by, e.g., \cite[Thm. XIII.58]{RS4};
	\item $ \sigma_{\mathrm{sing}}(H_U)= \sigma_{\mathrm{sing}}(H_V) = \emptyset $ by \cite[Thm. XIII.33]{RS4} and, therefore, $ \sigma_{\mathrm{ac}}(H_U) = \sigma_{\mathrm{ac}}(H_V) = [0,+\infty) $.
\end{itemize} 

In order to control the effect of the closed channel on the scattering, we also have to make some further assumptions on the spectra of $ - \Delta +  U $ and $ - \Delta + V $:

\begin{assumption} \label{boundst}
We assume that 
\begin{itemize}
	\item[a)] $H_U  $ has $N \geq 1$ negative simple eigenvalues $ E_0 < E_1< \cdots < E_{ N-1} < 0 $, with corresponding eigenvectors $ \psi_0, \psi_1, \ldots, \psi_{N-1} \in L^2(\R^3) $;
	\item[b)]  $H_V \geq 0 $ and zero is  neither an eigenvalue nor a  resonance.
	\item[c)]  $ \ker \lf( R_V^{1/2}(0) W \ri) = \lf\{ 0 \ri\} $. 
\end{itemize}
\end{assumption}

The above set of assumptions (in particular point a)) will be used only in next Theorem \ref{main} about the characterization of the singularities corresponding to Feshbach resonances. In fact, we are going to make such assumptions only for technical reasons and to simplify the discussion of the singularities of the scattering length. We do expect however the result to hold true even if, e.g., the eigenvalues of $ H_U $ have nontrivial multiplicities, but the proof will become much more involved in that case. 

We are now able to state the first result about the equation \eqref{geneigen3}. We denote by $ \EE $ the set of positive eigenvalues of $ \HH $, i.e.,
\beq
	\EE : = \sigma_{\mathrm{pp}}(\HH) \cap \R^+.
\eeq
We will prove $\EE$ is a discrete set containing only eigenvalues with finite multiplicity (Proposition \ref{pro:positive-eigenvalues}). Obviously, we are interested in the solutions of \eqref{geneigen3} having the form of generalized eigenfunctions or distorted plane waves and therefore we have to assume that $ k^2 \notin \EE $. In the following $ \mathscr{B} $ will stand for the set of continuous functions vanishing at infinity, i.e.,
	\beq
		\label{eq:Bfunctions}
		\mathscr{B} : = \Big\{ f \in C(\R^3) \: \Big| \: \lim_{x \to + \infty} f(\xv) = 0 \Big\}.
	\eeq

	\begin{proposition}[{Generalized eigenfunctions}]
		\label{pro:solution-scattering}
		\mbox{}	\\
		Let Assumption \ref{pot} hold true and let $ \lambda > 0 $ be fixed. Then, for any $ \kv \in \R^3 $ with $ k^2 \in (0,\la) \setminus \EE$ and
$k^2-\la \neq E_j$, for $j=0, \ldots, N-1$, equation \eqref{geneigen3} admits a unique continuous solution $ \varphi_\kv $, such that
$ \varphi_\kv - \eigenv \in \BB$. Furthermore, $ \varphi_\kv $ satisfies the asymptotic
		\beq
			\label{eq:varphi-asymptotics}
			\varphi_{\kv}(\xv) \underset{x \to +\infty}{\simeq} e^{i \kv \cdot \xv} + \ampe \frac{e^{i k x}}{x},
		\eeq
		with
		\beq \label{scatamp}
			\ampe = \tx\frac{1}{4 \pi} \mean{\phi_{V,\kvp}}{W\, R_U(k^2-\la)\, W}{\varphi_{\kv}} + \amp,
		\eeq
		where $ \amp $ is the scattering amplitude associated to the potential $ V $ (see \eqref{eq:ls-solution}).		
	\end{proposition}

	The above results sets the stage for the main result we prove in this paper. Indeed, the asymptotic \eqref{eq:varphi-asymptotics} shows that, by analogy with the one-particle case described above and the solution of the LS equation \eqref{eq:ls} given by \eqref{eq:ls-solution}, one can think of $ \ampe $ as the {\it effective scattering amplitude} in the open channel, taking into account the effects of the interaction with the closed one, as usually done in the physics literature  (see, e.g., \cite{Ga}).
	
	\n 
The next definition is thus the obvious consequence.
	
	\begin{definition}[Effective scattering length]
		\label{def:scattering-length}
		\mbox{}	\\
		We define the effective scattering length in the open channel as
		\beq
			a_{\mathrm{eff}}(\la) : = \lim_{k \to 0} \ampe,
		\eeq
		where $ \ampe $ is given by \eqref{scatamp}
	\end{definition}
	
Any resonance in the open channel is then associated to a singularity of the effective scattering length $ \aeff $. The next Theorem classifies all such singularities and the critical values of the control parameter $ \lambda $ yielding them.

	\begin{theorem}[Feshbach resonances]
		\label{main}
		\mbox{}	\\
		Let Assumptions \ref{pot} and \ref{boundst} hold true. Then,			
		\begin{enumerate}[(i)]
			\item there exists at least a critical value $ \la_0 \in (|E_0|, + \infty) $, such that the homogenous equation
				\beq
					\label{teo: homo}
					\lf(I - R_V (0) \, W\,  R_U(-\la_0)\, W \ri) \eta = 0,
				\eeq
				admits at least one non-trivial solution $ \eta_0 \in \mathscr{B} $. Furthermore, $ \aeff $ is continuous for $ \la \neq \la_0 $, with the possible exception of finitely many other points, and, if additionally
				\beq
					\label{eq: resonant condition}
					\braket{\phi_{V,0}}{   W  R_U(-\la_0) W \eta_0} \neq 0,	 \qquad	\forall \eta_0 \in \ker\lf(I - R_V (0) \, W\,  R_U(-\la_0)\, W \ri),
 				\eeq
				one has the expansion
				\beq
					\aeff = \frac{c_0}{\la - \la_0} + \OO(1),		\qquad		\mbox{as } \la \to \la_0,
				\eeq
				for some $ c_0 \neq 0 $;
			\item there exists $\de_0>0$ such that, if $ \lf\| W \ri\|_3 \leq \de_0$,
				then there are exactly $ N $ critical values $ \la_j $, $ j = 0, \ldots, N-1 $, satisfying
				\beq
					\la_0 > |E_0| > \la_1 > |E_1 | > \cdots > \la_{N-1} > |E_{N-1}|,
				\eeq
				such that the homogenous equation
				\beq
					\label{teo: homo perturbative}
					\lf(I - R_V (0) \, W\,  R_U(-\la_j)\, W \ri) \eta = 0,
				\eeq
				admits at least one non-trivial solution $ \eta_j \in \mathscr{B} $, for any $ j = 0, \ldots, N-1 $. Furthermore, $ \aeff $ is continuous for any $ \la \in [|E_{N-1}|, + \infty) $, $ \la \neq \la_j $, i.e., any further critical value $\la_j $, with $j\geq N$, is such that $ |E_{N-1}| > \la_j >0 $. Finally, if
				\beq
					\label{eq: resonant condition perturbative}
					\braket{\phi_{V,0}}{   W  R_U(-\la_j) W \eta_j} \neq 0, \qquad	\forall \eta_j \in \ker\lf(I - R_V (0) \, W\,  R_U(-\la_j)\, W \ri),
 				\eeq
				one has 
				\beq
					\aeff = \frac{c_j}{\la-\la_j} + {\mathcal O}(1),  \qquad		\mbox{as } \la \to \la_j,
				\eeq
				for some $c_j \neq 0 $.		
		\end{enumerate}
	\end{theorem}
	
	\begin{remark}[Resonant condition]
		\label{rem: resonant condition}
		\mbox{}	\\
		Condition \eqref{eq: resonant condition} (resp. \eqref{eq: resonant condition perturbative}) is not simply a technical assumption, but rather a key ingredient which plays the role of a necessary condition for the occurrence of the resonance at $ \la_0 $ (resp. $ \la_j $), at least if $ \ker\lf(I - R_V (0) \, W\,  R_U(-\la_0)\, W \ri) $ is one-dimensional, which we are going to assume in this discussion. In other words, if $
			\braket{\phi_{V,0}}{   W  R_U(-\la_0) W \eta_0} = 0 $,
			the scattering length remains bounded at $ \la_0 $ and no resonance is present. Concretely, this is due to the fact that $ c_0 $ is proportional to the above quantity. In fact, the explicit expression of the coefficients $ c_0 $ and $ c_j $ is provided in \eqref{eq:c0} and \eqref{eq:cj}, respectively, and, in principle, their values could be computed or numerically estimated, once the potentials $ U, V $ and $ W $ are given.
		
		In order to understand the role played by such a condition, it is convenient to consider the analogy with one-particle Schr\"{o}dinger operators: it is very well known that an operator of the form $ - \Delta + v $ in three-dimensions has a zero-energy resonance if the operator $ |v|^{1/2} (-\Delta)^{-1} |v|^{1/2} $ has the eigenvalue $ 1 $ and, denoting by $ \phi_0 \in L^2(\R^3) $ a corresponding eigenvector, it must be 
		\beq
			\label{eq: resonant condition conventional}
			\braket{|v|^{1/2}}{\phi_0} \neq 0.
		\eeq
		Correspondingly, the scattering length diverges. If the above condition is not met, then $ 0 $ is not a resonance of $ - \Delta + v $ but rather a zero-energy eigenvalue, and the scattering length remains finite.
		
		Now, condition \eqref{eq: resonant condition} is precisely the analogue of \eqref{eq: resonant condition conventional} for our matrix Hamiltonian. This is also made apparent by the fact that $  W  R_U(-\la) W $ plays the role of an effective potential in the open channel and \eqref{eq: resonant condition} reduces to \eqref{eq: resonant condition conventional} for $ V = 0 $.	
	\end{remark}
	
	\begin{remark}[Number of critical points]
		\label{rem:Npoints}
		\mbox{}	\\
		Without the smallness hypothesis on $ W $, we can not prove that there are {\it at least} $ N $ critical points, due to a possible accidental degeneracy occurring to a compact operator, which plays a key role in the analysis (see Proposition \ref{crit} and Remark \ref{rem:Npoints-homo} for further details). Of course such a degeneracy does not typically occurs, since any infinitesimal change of the potentials $ U, V $, or $ W $, would break it and we expect at least $ N $ critical values as in case (ii). The further singular points mentioned at point (i) of the above Theorem, where $ a(\lambda) $ is not continuous, become the $ \la_j$'s localized at point (ii), thanks to the smallness of $ W $. 
	\end{remark}

	It is well known \cite[Def. 2]{hs} that in the one-body case the scattering length associated to a potential $ v $ can be expressed as
\beq
	a_v = \tx\frac{1}{4\pi}  \mean{v^{1/2}}{\lf( I +|v|^{1/2} (-\Delta)^{-1} v^{1/2} \ri)^{-1}}{|v|^{1/2}}
\eeq
where we have conventionally set $ v^{1/2} : = \sgn(v) |v|^{1/2} $. Hence, any divergence of the scattering length is associated to a nontrivial kernel of the operator $  I +|v|^{1/2} (-\Delta)^{-1} v^{1/2} $, i.e., a nontrivial solution of the zero-energy Schr\"{o}dinger equation $ (-\Delta + v) \psi = 0 $, which in turn implies the presence of a zero-energy resonance, provided the condition \eqref{eq: resonant condition conventional} is satisfied. The same behavior holds true in our case, namely there is a Feshbach resonance if and only if there is a non-$L^2$ solution of the zero-energy Schr\"{o}dinger equation (see \cite{CF} for a low-energy expansion of the resolvent).

	\begin{corollary}[Zero-energy equation]
 		\label{zerores}
 		\mbox{}	\\
		Under the same assumptions of Theorem \ref{main}, if $\la=\la_j$, then there exists a distributional solution of the zero-energy equation $\HH \Psi = 0$.
	\end{corollary}

Before dealing with the proof of the main results above, we present in next Section \ref{sec:separable} a simple solvable model in which the effective scattering length can be derived almost explicitly and the corresponding resonances easily found, i.e., the case of an inter-channel interaction given by a {\it separable} or {\it one-rank potential} of the form $ W  = \ket{w}\bra{w} $. Such an example will prove to be useful to present and mimic the  overall strategy of the proof, which we will discuss in the rest of the paper.

\section{Separable Potentials}
\label{sec:separable}

We consider here a toy model built in terms of separable potentials, i.e., we replace the potential $ W $ with a projector onto a function $ w \in I_3 $, i.e., only in this Sect. we set
\beq
	W = \ket{w} \bra{w}.
\eeq
Then,  equation  \eqref{geneigen3} takes the form:
\beq
	\varphi_\kv  = \phi_{V,\kv} + \meanlrlr{w}{R_U(k^2-\la)}{w} \braket{w}{\varphi_\kv} R_V(k^2) w.
\eeq
Taking the inner product of the above equation with $w$, we obtain
\bdm
	\braket{w}{\varphi_\kv} = \meanlrlr{w}{R_U(k^2-\la)}{w} \braket{w}{\varphi_\kv} \meanlrlr{w}{R_V(k^2)}{w} + \braket{w}{\phi_{V,\kv}}
\edm
and thus
\beq
	\label{eq:separable1}
	\varphi_\kv  =   \phi_{V,\kv} +  \frac{ \braket{w}{ \phi_{V,\kv}}\meanlrlr{w}{R_U(k^2-\la)}{w}}{1 -  \meanlrlr{w}{R_V(k^2)}{w}\meanlrlr{w}{R_U(k^2-\la)}{w}}  R_V(k^2) w.
\eeq
In order to recover the asymptotic \eqref{eq:varphi-asymptotics}, it now suffices to use the resolvent identity and the standard large distance expansion for $R_0(k^2) w$ (see, e.g., \cite{I}), i.e., 
\beq
	\lf(R_0(k^2) f\ri)(\xv) \underset{x \to + \infty}{=}  \frac{e^{i k x}}{4\pi x} \int_{\R^3} \diff \xvp \: e^{-i\kvp \cdot \xvp}  f(\xvp) + o(x^{-1}) = \braket{\phi_{0,\kvp}}{f} \frac{e^{i \kv \cdot \xv}}{4\pi x} + o(x^{-1}) ,
\eeq
where $ \kvp = k \hat{\xv} $, to get
\bmln{
	\lf(R_V(k^2) w\ri)(\xv) =  \lf(R_0(k^2) \lf(1 - VR_{V}(k^2)\ri) w\ri)(\xv) \\
	\underset{x \to + \infty}{=}  \braketr{\phi_{0,\kvp}}{\lf(1 - VR_{V}(k^2)\ri) w} \frac{e^{i k x}}{4\pi x} + o(x^{-1}) .
}
On the other hand, the LS equation can be rewritten as $\;
	\lf( I - R_{V}(k^2)V \ri) \phi_{0,\kvp} = \phi_{V,\kvp} \;$ 
and therefore
\beq
	\lf(R_V(k^2) w\ri)(\xv) \underset{x \to + \infty}{=}  \braket{\phi_{V,\kvp}}{w} \frac{e^{i k x}}{4\pi x} + o(x^{-1}).
\eeq
Hence, combining \eqref{eq:separable1}, the above asymptotics and the expansion \eqref{eq:ls-solution} for $ \phi_{V,\kv} $, we get
\bdm
	\varphi_\kv (\xv)  \underset{x \to + \infty}{=} e^{i \kv \cdot \xv} + \ampe  \frac{e^{i k x}}{x} + o(x^{-1}),
\edm
with
\beq
	\ampe =  \amp  + \frac{1}{4\pi} \frac{ \braket{w}{ \phi_{V,\kv}}  \braket{ \phi_{V,\kvp}}{w} \meanlrlr{w}{R_U(k^2-\la)}{w}}{1 -  \meanlrlr{w}{R_V(k^2)}{w}\meanlrlr{w}{R_U(k^2-\la)}{w}}.
\eeq
To compute the scattering length, we have then to take the limit  $ k\to 0$, obtaining
\beq
	\label{eq:aeff-separable}
	\aeff = a_V + \frac{1}{4\pi} \frac{\lf| \braket{w}{ \phi_{V,0}}\ri|^2  \meanlrlr{w}{R_U(-\la)}{w}}{1 -  \meanlrlr{w}{R_V(0)}{w}\meanlrlr{w}{R_U(-\la)}{w}}.
\eeq

In order to find the singularities of the scattering length, one has to investigate the denominator of the second term on the r.h.s. of the expression above, as $ \lambda $ varies in $ \R^+ $, since $ a_V $ is finite by assumption. We also note that $\meanlrlr{w}{R_V(0)}{w} \geq 0$ by Assumption 2, point (b), and $ \meanlrlr{w}{R_V(0)}{w} < + \infty$, which can be proved along the lines of Remark \ref{rem:k02}.

Let us  introduce the following  further assumptions: 
\beq
	 \meanlrlr{w}{R_V(0)}{w} >0,
\eeq
which is a rephrasing of Assumption \ref{boundst}, point c), and
\beq
	\label{eq:hp-R}
	\braket{w}{\psi_j} \neq 0,		\qquad		\mbox{for any } j = 0, \ldots, N-1.
\eeq
Observe that condition \eqref{eq:hp-R} is the analogue of the resonant condition \eqref{eq: resonant condition}. By equation \eqref{eq:aeff-separable} one sees that a singularity of $ \aeff $ does occur only for $\lambda$ solution of the equation 
\beq
	\label{eq:eq-lambda}
	F(\lambda) = \frac{1}{\meanlrlr{w}{R_V(0)}{w}}.
\eeq
where we have denoted $ F(\lambda) : = \meanlrlr{w}{R_U(-\la)}{w} $. By the spectral resolution of $H_U$ we have
\beq
	F(\lambda) = \sum_{j = 0}^{N-1} \frac{1}{E_j + \lambda} \lf| \braket{\psi_j}{w} \ri|^2 + \int_{0}^{+\infty} \diff \mean{w}{E(\mu)}{w} \: \frac{1}{ \mu + \lambda},
\eeq
where the second term is continuous and monotonically decreasing and the first one has singularities only at $ \lambda = - E_j $, $ j = 0, \ldots, N-1 $. More precisely,
\bdm
	\lim_{\lambda \to - E_j^{\pm}} F(\lambda) = \pm \infty.
\edm
Moreover 
\bdm
	\lim_{\lambda \to + \infty} F(\lambda) = 0,		
\edm
and $ F(\lambda) $ is decreasing in $ \lambda $: taking the derivative w.r.t. $ \lambda$ for $ \lambda \neq - E_j $, $j = 1, \ldots, N-1 $, we get
\beq
	\label{eq:Fprime}
	F^{\prime}(\lambda) = - \meanlrlr{w}{R^2_U(-\la)}{w} = - \lf\| R_U(-\la) w \ri\|_2^2 < 0.
\eeq
Hence, we obtain that if $\lim_{\lambda \rightarrow 0}F(\lambda)\leq 1/\meanlrlr{w}{R_V(0)}{w} $ then equation \eqref{eq:eq-lambda} admits  $ N $ solutions $ \lambda_0, \ldots, \lambda_{N-1} \in \R^+$, such that
\beq
	\label{eq:lambdi}
	0 < \lf|E_{N-1}\ri| < \lambda_{N-1} < \cdots < \lf|E_1\ri| < \lambda_1 < \lf|E_0\ri| < \lambda_0.
\eeq
On the other hand, if $\lim_{\lambda \rightarrow 0}F(\lambda)> 1/\meanlrlr{w}{R_V(0)}{w} $ then there is an additional solution $ \lambda_N $, such that 
\bdm
	0 < \lambda_N < |E_{N-1}|.
\edm

\begin{center}
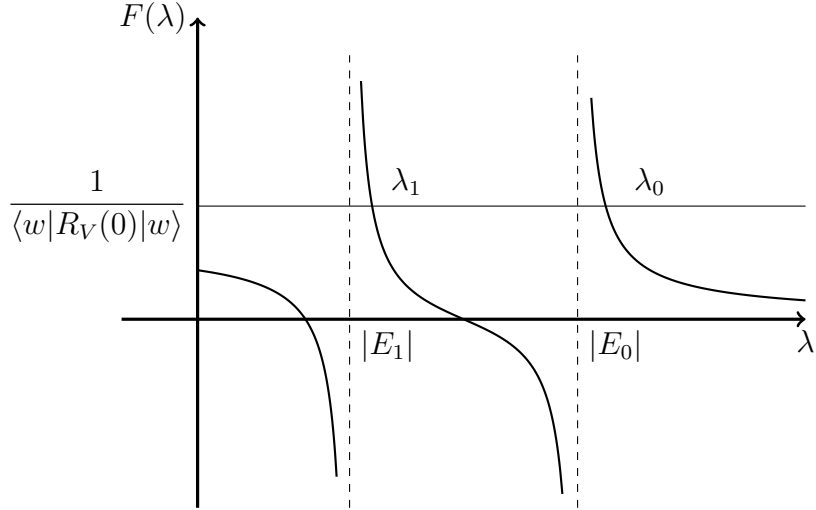
\begin{figure}
\label{fig:F}
\begin{tikzpicture}[scale=1]
\draw[->, very thick] (-1,0) to (8,0) ;
\draw[->, very thick] (0,-2.5) to  (0,4) ;
\node [below] at (8,0) {\scalebox{1.0}{$\la$}};
\node [left] at (0,1.5) {\scalebox{1.0}{$\dfrac{1}{\langle w| R_V(0) |w\rangle}$}};
\node[left] at (0,4) {\scalebox{1.0}{$F(\la)$} };
\draw[domain=5.18:8, samples=100, thick] plot[smooth] (\x, {(\x-3.5)/((\x-5)*(\x-2)) } ); 
\draw[domain=2.15:4.8, samples=100, thick] plot[smooth] (\x, {(\x-3.5)/((\x-5)*(\x-2)) } ); 
\draw[domain=0:1.83, samples=100, thick] plot[smooth] (\x, {1+(\x-3.5)/((\x-5)*(\x-2)) } );
\draw[thin, dashed] (5,-2.5) to (5,3.5);
\draw[thin, dashed] (2,-2.5) to (2,3.5);
\draw (0,1.5) to (8,1.5);
\node[above right ] at (5.6,1.5) {$\la_0$ };
\node[above right] at (2.4,1.5) {$\la_1$ };
\node[below right] at (5,0) {$|E_0|$ };
\node[below right] at (2,0) {$|E_1|$ };
\end{tikzpicture}
\caption{Qualitative behavior of the function $ F(\la) $, $ \la \in \R^+ $.}
\end{figure}
\end{center}

Therefore, we conclude that there exist at least $ N $ points $ \lambda_j \in \erre^+$, $ j = 0, \ldots, N-1 $, satisfying \eqref{eq:lambdi}, so that
\beq
	\aeff = \frac{c_j}{\lambda - \lambda_j} + \OO(1), 
\eeq
since by \eqref{eq:Fprime} the first derivative of $ F(\lambda) $ does not vanish at $ \lambda_j $. The explicit value of the constant $ c_j $ can be easily obtained from \eqref{eq:aeff-separable} and \eqref{eq:Fprime}:
\beq
	c_j  = \frac{1}{4\pi}  \frac{\lf| \braket{w}{ \phi_{V,0}}\ri|^2 \meanlrlr{w}{R_U(-\la_j)}{w}}{\meanlrlr{w}{R_V(0)}{w}\lf\| R_U(-\la_j) w \ri\|_2^2}.
\eeq
We have thus shown the result expressed in theorem \ref{main} (Feshbach resonances) in the special case of an inter-channel interaction given by a separable potential.

\section{Notation and Technical Lemmas}
\label{sec:notation}

We collect here some notation and useful technical results, which will be used in the rest of the paper. 

Given the important role played in the following, we recall the definition of weighted Hilbert spaces: set $ \lan x \ran : = \sqrt{1 + x^2} $ for short, then, for any $ s \geq 0 $, we define 
\beq
	\label{eq:wnorm}
	\lf\| f \ri\|_{L^2_s} : = \lf\| \lan x \ran^s f \ri\|_2.
\eeq
The closure of $ C^{\infty}_0 $ w.r.t. the above norm is denoted by $ L^2_s $. The weighted Sobolev space $ H^2_s $, $ s \geq 0 $, is defined analogously as the closure of $ C^{\infty}_0 $ w.r.t. the norm
\beq
	\label{eq:wSnorm}
	\lf\| f \ri\|_{H^2_s} : = \lf\| f \ri\|_{L^2_s} + \lf\| \Delta f \ri\|_{L^2_s}.
\eeq
The conventional Sobolev spaces $ H^p $, $ p \in \R $ can be defined via Fourier transform as the closure of $ C^{\infty}_0(\R^3) $ w.r.t. the norms 
\beq
	\lf\| f \ri\|_{H^p} : = \big\| \lan k \ran^p \hat{f} \big\|_2,
\eeq
 where we use the following convention for the Fourier transform
\beq
	\hat{f}(\kv) : = \frac{1}{(2\pi)^{3/2}} \int_{\R^3} \diff \xv \: e^{-i \kv \cdot \xv} f(\xv).
\eeq
By the properties of the Fourier transform and a simple exchange of the role of $ \xv $ and $ \kv $, one easily gets
\beq\label{LsHs}
	 \lf\| f \ri\|_{L^2_s}^2 = \int_{\R^3} \diff \xv \: \lf(1 +  x^2 \ri)^{s} \lf| f(\xv) \ri|^2 = \int_{\R^3} \diff \xv \: \lf( 1 + x^2 \ri)^{s} \Big| \widehat{\hat{f}}(-\xv) \Big|^2 = \big\| \hat{f} \big\|_{H^s}^2.
\eeq
Hence, we obtain the useful identity
\beq
	\label{eq:wSnorm-identity}
	\lf\| f \ri\|_{H^2_s} = \big\| \hat{f} \big\|_{H^s} + \big\| k^2 \hat{f} \big\|_{H^s}.
\eeq

We now recall some classical results on spectral and scattering theory mostly taken from \cite{A,I}, which will be used in the proofs. Recall the definition \eqref{eq:Bfunctions} of the space $ \mathscr{B} $ of continuous functions vanishing at $ \infty $. We also denote by $ {\mathcal B}(L^2) $ the space of bounded linear operators on $ L^2 $ and, more in general, $ \mathcal{B}(X,Y) $ stands for the space of continuous linear transformations between two Banach spaces $ X $ and $ Y $. Similarly, $ \LL^{\infty}(X,Y) $ is the space of compact operators from $ X $ to $ Y $ and $ \LL^{\infty}(X) : = \LL^{\infty}(X,X) $.

	\begin{remark}[$ I_n $ and compactness]
		\label{rem:compact}
		\mbox{}	\\
		By Sobolev embedding of $ H^2 $ in the space of continuous function, one can easily realize that the multiplication operator by a potential $ Z \in I_n $ belongs to $ \LL^{\infty}(H^2_s, L^2_{s+n}) $, for any $ s > 0 $.
	\end{remark}

We also state some  properties of the operator $ R_0(k^2) Z $ on $ \mathscr{B} $, where $ Z $ is a suitable Ikebe potential. In the following the parameter $ \delta > 0 $ is the one appearing in the definition of the Ikebe class (see Definition \ref{def:ikebe}) and thus it is a characteristic of the corresponding potential.

	\begin{lemma}[\mbox{\cite[Lemmas 3.1, 3.2, 4.1, 4.2 \& 4.6]{I}}] \label{ike}
		\mbox{}	\\
		Assume $Z\in I_2$ and consider the operator $T_Z (k) := R_0 (k^2) Z$ with integral kernel given by
		\beq
			T_Z \lf(\xv , \yv ; k\ri) = \frac{e^{i k|\xv-\yv| }}{4\pi|\xv-\yv|}  Z(\yv).
		\eeq
		Then,
		\begin{itemize}
			\item[a)] $T_Z (k) \in \LL^{\infty}(\BB) $;
			\item[b)] if $ f \in \mathscr{B} $, then 
				\beq
					 \lf( T_Z(k) f \ri) (\xv) \underset{x \to +\infty}{=}  \OO(x^{-\delta});
				\eeq
			\item[c)] if $ f \in \mathscr{B} $ and $ f(\xv) \underset{x \to +\infty}{=} \OO (x^{-1} )  $, then 
				\beq
					\lf( T_Z (k) f\ri)  (\xv) = \f{e^{i k x}}{4 \pi x} \int_{\RE^3} \diff \yv \: e^{-i \kvp \cdot \yv } f(\yv) + \OO(x^{-1-\de/2}) ;
				\eeq
			\item[d)] the map $ k \to T_Z (k)$ from $ \R^+ $ to $ \mathcal{B}(\mathscr{B}) $ is continuous in the uniform topology.
		\end{itemize}
	\end{lemma}
	
	Next we recall some useful properties of the free resolvent in weighted Hilbert spaces.

	\begin{lemma}[\mbox{\cite[Thm. 4.1]{A}}]
		\label{agmon1}
		\mbox{}	\\
		For any $k^2> 0$, $R_0 (k^2) \in {\mathcal B} (L^2_s, H^2_{-s}) $, for any $ s > 1/2$. Furthermore, for any $f \in L^2_s$, with $s>1/2$, the function $ u= R_0 (k^2) f$ solves
		\beq
			\lf(-\Delta - k^2 \ri) u =  f
		\eeq
		in distributional sense and the following identity holds:
		\beq
			\Im \braketl{R_0 (k^2)f}{f} = \f{\pi}{ 2 k} \int_{S^2(k)} \diff \sigma(\kv) \: \big| \hat f (\kv)  \big|^2,
		\eeq
		where $ S^2(k) $ is the sphere of radius $ k $ and $\diff \sigma(\kv) $ is the surface  measure.
	\end{lemma}

	\begin{remark}[Extension to $ k = 0 $]
		\label{rem:k01}
		\mbox{}	\\
		The same results of Lemma \ref{agmon1} hold for $k =0$ under the stronger assumption $s>1$: see, e.g., \cite{Y} for the proof that  $R_0 (0) \in {\mathcal B} (L^2_s , L^2_{-s})$, while the improvement to $H_{-s}^2$ is standard 
		\[
			\lf\| R_0 (0) f \ri\|_{H^2_{-s}} = \lf\| R_0(0) f \ri\|_{L^2_{-s}} + \lf\| \Delta R_0 (0) f \ri\|_{L^2_{-s}} \leq C \lf\| f \ri\|_{L^2_{s}}
+ \| f\|_{L^2_{-s}} \leq C  \| f\|_{L^2_{s}}.
		\]
		Notice also that $ R_0 (k^2) $, $ k \geq 0 $, is a compact operator from $ L^2_s \to L^2_{-s} $ (see \cite{Y}).
	\end{remark}

	\begin{lemma}[\mbox{\cite[Thm. 3.2]{A}}] 
		\label{agmon2}
		\mbox{}	\\
		Let $f\in H^s$, $s>1/2$,  and let the spherical average of $f$ vanish on the sphere of radius $ k > 0 $. Set for any $ \alpha_j \in \N $, $ j = 1,2,3 $, such that $0 \leq \sum_j \alpha_j \leq 2$,
		\[
			g_{\alpha} (\xv) := \frac{x_1^{\alpha_1} x_2^{\alpha_2} x_3^{\alpha_3} f(\xv)}{x^2-k^2}.
		\]
		Then, $g_\al \in H^{s-1} \cap L^1_{\mathrm{loc}}$ and  $\lf\| g_\al \ri\|_{ H^{s-1}} \leq C \|f\|_{H^s}$.
	\end{lemma}

	Finally, we consider the resolvent $ R_{V}(k^2) $ (recall that $ V \in I_4 $ by Assumption \ref{pot}, point (b)). Exploiting the fact that $ V $ has no non-negative eigenvalues, we have

	\begin{lemma}[\mbox{\cite[Thm. 4.2]{A}}] 
		\label{agmon3}
		\mbox{}	\\
		For any $k^2> 0$, $R_V (k^2) \in {\mathcal B} (L^2_s , H^2_{-s})$, for any $s>1/2$. Furthermore, for any $f \in L^2_s$, with $s>1/2$, the function $u= R_V (k^2)f$ solves 
		\beq
			(-\Delta +V- k^2 ) u =  f
		\eeq
		in distributional sense and the resolvent identity holds, i.e.,
		\beq
			R_V (k^2)f = R_0 (k^2)f - R_0 (k^2)\, V \, R_V (k^2)f.
		\eeq
	\end{lemma}

	\begin{remark}[Extension to $ k = 0 $]
		\label{rem:k02}
		\mbox{}	\\
		For $ k = 0 $ the above result does not apply, but one has nevertheless
		\beq \label{bujo}
			R_V (0): L^2_s \to H^2_{-s}, \qquad	\mbox{for any }  s>1.
		\eeq
		Indeed, since $R_0 (0) \in {\mathcal B} ( L^2_{-s},H^2_s)$, for any $ s > 1 $, we have that $V R_0 (0) \in {\mathcal B}( L^2_{s})$ and it is in fact a compact operator on $ L^2_s $, by Remark \ref{rem:compact} and Assumption \ref{pot}, point b). 
Therefore, by Freedom alternative
and Assumption \ref{boundst}, point b), $(I + V R_0 (0) )^{-1} \in {\mathcal B}( L^2_{s}) $: absence of zero energy eigenvalues or resonances for $ V $ guarantees that the equation $ (I + V R_0 (0) ) f = 0 $ admits only the trivial solution. In  conclusion, the resolvent identity
$R_V (0) = R_0 (0) (I + V R_0 (0) )^{-1}$ implies the claim via Remark \ref{rem:k01}.
	\end{remark}


\section{Generalized eigenfunctions}
\label{sec:eigenfunctions}

In the first  two propositions we relate of the positive eigenvalues of $\HH$ and the solutions of the homogeneous equation associated to \eqref{geneigen3}. Let us then set 
\beq
	\EE : = \sigma_{\mathrm{pp}}(\HH) \cap \R^+.
\eeq

	\begin{proposition}[Positive eigenvalues of $ \HH $]
		\label{pro:positive-eigenvalues}
		\mbox{}	\\
		The set $\EE$ is discrete and its only accumulation points can be $0$ or $+ \infty$. Furthermore, each eigenvalue in $ \EE $ has finite multiplicity.
	\end{proposition}

	\begin{remark}[Positive eigenvalues]
		\label{rem: positive eigen}
		\mbox{}	\\
		We expect this characterization of positive eigenvalues not to be optimal but it is
sufficient for our scopes. For instance, one could adapt Kato's argument as presented in \cite{RS4}
and prove with some lengthy calculations that there are no positive eigenvalues in $(\la,+\infty)$.
On the contrary, we expect that there may be positive eigenvalues in $(0,\la)$. Indeed, in the unperturbed case, i.e., for $W=0$, there may be embedded eigenvalues, depending on the relation between $\{E_j\}_{j = 0, \ldots, N-1}$ and $\la$. When $W\neq 0$, it is expected  that such eigenvalues turn into complex resonances (see, e.g., \cite{ccf1, ccf2, frr}) for the analysis of this phenomenon in some solvable models). As discussed in \cite{ahs} in the one-body case, however, there may be persistence of embedded eigenvalues for some special perturbations.
	\end{remark}
	
	\begin{proof}
		Let $\Psi_{\kv}=(\zeta_{\kv}, \theta_{\kv}) \in H^2 \oplus H^2$ be a solution of  the eigenvalue equation $\HH \Psi_{\kv} = k^2 \Psi_{\kv}$, for $k^2>0$. Equivalently, we have that $\zeta_{\kv}, \theta_{\kv} \in H^2$ are  solutions of the system \eqref{geneigen}.  
Following the line of \cite[Thm. 3.1]{A}, we preliminary prove that there exists $\ve>0$ such that 
		\beq 
			\label{sarda}
		 	\lf\| \zeta_{\kv} \ri\|_{H^2_\ve} \leq C \, \lf(\lf\| \zeta_{\kv} \ri\|_{2} + \lf\| \theta_{\kv} \ri\|_{2}\ri), \qquad	\lf\| \theta_{\kv} \ri\|_{H^2_\ve} \leq C \, \lf(\lf\| \zeta_{\kv}\ri\|_{2} + \lf\| \theta_{\kv}\ri\|_{2}\ri) .
		\eeq
		The system \eqref{geneigen} can be rewritten as
		\be 
			\begin{cases} 
  				\zeta_{\kv}   = (k^2 -i) R_V(i)   \zeta_{\kv} - R_V(i) W\,   \theta_{\kv}, \\
  				\theta_{\kv}  =  (k^2 -i) R_U(i) \theta_{\kv} - R_U(i)  W\, \zeta_{\kv}.
			\end{cases}
		\ee
		Then, by Sobolev embedding and the boundedness of $(-\Delta +1) R_V (i)$ and $(-\Delta +1) R_U (i)$, we have for any $ s > 3/2 $ that
		\beq \label{triglia}
			\begin{cases} 
 				\lf\| \zeta_{\kv} \ri\|_{H^2} \leq   C\, \lf\|   \zeta_{\kv} \ri\|_{2}+ \lf\| W \ri\|_{2} \lf\|  \theta_{\kv} \ri\|_{L^\infty}    \leq C \, \lf(\lf\|   \zeta_{\kv} \ri\|_{2} + \lf\|\theta_{\kv}  \ri\|_{H^s} \ri), \\
				\lf\|  \theta_{\kv} \ri\|_{H^2} \leq C\,  \lf\|\theta_{\kv} \ri\|_{2} +C\,  \lf\|W\ri\|_{2}  \lf\|\zeta_{\kv} \ri\|_{L^\infty} \leq C\, \lf(   \lf\|\theta_{\kv} \ri\|_{2} +  \lf\|\zeta_{\kv} \ri\|_{H^s}\ri).
			\end{cases}
		\eeq
		Let us now pick some $ s \in (3/2,2) $: then, for any $  \ve >0$ there exists a finite $ c_{\ve} >0$, such that $ \lf\| f \ri\|_{H^s} \leq \ve  \lf\| f  \ri\|_{H^2} + c_{\ve}  \lf\| f \ri\|_2 $, and from \eqref{triglia} we obtain 
		\beq \label{spigola}
			\begin{cases} 
				\lf\| \zeta_{\kv} \ri\|_{H^2}  \leq C \, \lf(\lf\|   \zeta_{\kv} \ri\|_{2} + \lf\|\theta_{\kv}  \ri\|_{2} \ri), \\
				\lf\|  \theta_{\kv} \ri\|_{H^2} \leq C\, \lf(   \lf\|\theta_{\kv} \ri\|_{2} +  \lf\|\zeta_{\kv} \ri\|_{2}\ri).
			\end{cases}
		\eeq
		Thanks to Assumption \ref{pot} and Remark \ref{rem:compact}, the multiplication operators by the potentials $V$, $U$ and $W$ belong to $ \LL^{\infty}(H^2, L^2_{1+\ve}) $. In particular, 
 using also \eqref{spigola}, we have 
		\bdm
			\lf\| W \theta_{\kv} \ri\|_{L^2_{1+\ve} } \leq C\, \lf\| \theta_{\kv} \ri\|_{H^2 } \leq C \, \lf( \lf\|   \zeta_{\kv} \ri\|_{2} + \lf\| \theta_{\kv} \ri\|_{2} \ri).
		\edm
Using Fourier transform, \eqref{geneigen} can be written as
		\beq \label{tonno}
			\begin{cases} 
				(p^2-k^2) \widehat{ \zeta_{\kv}} (\pv)  = - \widehat{ W\,   \theta_{\kv}} (\pv)  -\widehat{   V \zeta_{\kv}} (\pv) =:\widehat{g_{\kv,1}}(\pv),   \\
				(p^2-k^2) \widehat{ \theta_{\kv}}(\pv)  = - \widehat{ W\,   \zeta_{\kv}} (\pv)  -\widehat{U \theta_{\kv}}(\pv) - \lambda \widehat{\theta_{\kv}}(\pv) =: \widehat{g_{\kv,2}}(\pv). 
			\end{cases}
		\eeq
		Due to the above remark, $g_{\kv,i} \in L^2_{1+\ve}$ and $\widehat{g_{\kv,i}} \in H^{1+\ve}$, $i=1,2$. Moreover,  \eqref{tonno} implies that
the trace of $\widehat{g_{\kv,i}} $ on the sphere $|\pv |=k$ vanishes (see also \cite{A}). By Lemma \ref{agmon2}, it follows that 
		\beq
			\widehat{\zeta_{\kv}} (\pv) = \frac{ \widehat{g_{\kv,1}}(\pv)}{ p^2-k^2 } 
		\eeq
		belongs to $H^\ve$ and, in addition,  $\pv^{\bm{\alpha}} \widehat{\zeta_{\kv}} (\pv)$ belongs to   $ H^\ve $ for any multi-index $ \bm{\alpha} $ such that $| \bm{\al} |\leq 2 $. Therefore, by \eqref{eq:wSnorm-identity} and again by Lemma \ref{agmon2}, $ \zeta_{\kv}  \in H^2_\ve$ and
		\beq
			\lf\|  \zeta_{\kv}  \ri\|_{ H^2_\ve} \leq C \, \lf\| \widehat{g_{\kv,1}} \ri\|_{H^{1+\ve}} = C \,\lf\| g_{\kv,1} \ri\|_{L^2_{1+\ve} } \leq C\,  \lf( \lf\| \zeta_{\kv} \ri\|_{L^2} + \lf\| \theta_{\kv} \ri\|_{L^2} \ri).
		\eeq
		The same estimate holds for $\theta_{\kv}$ and therefore \eqref{sarda} is proven.

		Let us now prove that there can be only a finite number of positive eigenvalues in the interval $ (a,b) $, with $0< a < b <+\infty$, and that their multiplicity is finite. Suppose by absurd that this is false and let $\{ \Psi_j \}_{j \in \N} $ a sequence of orthonormal eigenvectors. Then by \eqref{sarda} and Rellich's criterion \cite[Thm. XIII.65]{RS4}, the sequence $\{ \Psi_j \}_{j \in \N}$ lies in a compact subset of $ L^2 $ and we could extract a convergent subsequence. But this
is absurd since the eigenvectors $\Psi_j $ are orthogonal.
\end{proof}

In the next proposition  we study the homogeneous equations associated to \eqref{geneigen3}, i.e.,
\beq 
	\label{eqfonhom}
	\lf( I- D_{\kv} \ri) \zeta_{\kv} = 0
\eeq
where we have set
\beq
	D_{\kv}:= R_V (k^2) \, W\,  R_U(k^2-\la)\, W\,.
\eeq

	\begin{proposition}[Solutions of \eqref{eqfonhom}]
		\label{poseigen}
		\mbox{}	\\
		Let Assumption \ref{pot} hold and let $k^2 \in (0,\la)$, with $k^2-\la \neq E_j$, $j=0, \ldots, N-1$. Then, $k^2\in \EE$ if and only if there exists a non-trivial solution $\zeta_{\kv} \in H^2 \subset \BB $ of \eqref{eqfonhom}.
	\end{proposition}
	
	\begin{proof}
		Let $\zeta_{\kv}\in \BB$ be a solution of \eqref{eqfonhom}. We first  prove that  $\zeta_{\kv}  \in H^2$ and
		\beq \label{penna}
			(-\Delta+V -k^2) \zeta_{\kv} = W\, R_U(k^{2}-\la)\, W\, \zeta_{\kv}.
		\eeq
		Let us denote $ \chi_{\kv} :=  W\, R_U(k^{2}-\la)\, W\, \zeta_{\kv}$, so that  $\zeta_{\kv} = R_V (k^2) \chi_{\kv} $. Then, by Assumption \ref{pot} and boundedness of $R_U(k^{2}-\la) $,  we have that $ \chi_{\kv} \in L^2_s $, for some $s>3/2$. By Lemma \ref{agmon3} we have that $ \zeta_{\kv}\in H^2_{-s}$ and  \eqref{penna} holds at least in distributional sense. Moreover, we note that the r.h.s. of \eqref{penna} belongs to $L^2$ and therefore  we also have  $(-\Delta +V -k^2) \zeta_{\kv}  \in L^2$.

		Let us now  show that  $\zeta_{\kv}  \in L^2$. We notice that we can not straightforwardly bootstrap the argument to conclude that $ \zeta_{\kv}\in L^2_s$, for any $ s\in \RE$, as it is done in \cite{A}, because of the non-locality of the effective potential $ W\, R_U(k^{2}-\la)\, W$. However, using the resolvent identity (see Lemma \ref{agmon3}), we can cast \eqref{eqfonhom} in the following form:
		\be
			\zeta_{\kv} - R_0 (k^2) \, W\,  R_U(k^2-\la)\, W\, \zeta_{\kv} + R_0 (k^2) V \, R_V (k^2)\, W\,  R_U(k^2-\la)\, W\, \zeta_{\kv}=0
		\ee
		or, equivalently, 
		\beq
			\zeta_{\kv} - R_0 (k^2) f_{\kv} =0, \qquad
			 f_{\kv} := \lf( I -  V \, R_V (k^2)\ri) \chi_{\kv}.
		\eeq
		By  the assumption on $V$, we have $  f_{\kv} = \chi_{\kv} - V \, R_V (k^2) \chi_{\kv} = \chi_{\kv} -V \zeta_{\kv} \in L^2_s$.  
Hence, we can apply Lemma \ref{agmon1}, to get
		\begin{multline}\label{trnulla}
			\f{\pi}{ 2 k} \int_{S^2(k)}  \!\! \diff \sigma(\pv) \: \big|\widehat{ f_{\kv}} (\pv)  \big|^2  = 
 			\Im \braketl{R_0 (k^2)f_{\kv}}{f_{\kv}}  =
  			\text{Im} \braketr{\zeta_{\kv}}{ (I- V \, R_V (k^2) )   W\, R_U(k^{2}-\la)\, W\, \zeta_{\kv}} \\
			= \text{Im} \lf( \mean{W\zeta_{\kv}}{ R_U(k^{2}-\la)}{ W\, \zeta_{\kv}} - \mean{ \zeta_{\kv}}{V}{ \zeta_{\kv}} \ri) = 0
		\end{multline}
		where, in the last line, we have used equation \eqref{eqfonhom}.

		Taking into account  that $f \in L^2_s$ is equivalent to $\hat{f} \in H^s$, by (\ref{trnulla}) and Lemma \ref{agmon2}, we obtain
that $  (p^2 - k^2)^{-1} \widehat{f_{\kv}} = \hat{\zeta}_{\kv}\in H^{s-1} \cap L^1_{\mathrm{loc}}$. This means that $\zeta_{\kv} = R_0 (k^2) f_{\kv}  \in L^2_{s-1}$, with $s-1>1/2$, and then, in particular, $\zeta_{\kv} \in L^2$. Thus, by $\zeta_{\kv} \in L^2$ and $(-\Delta +V -k^2) \zeta_{\kv}  \in L^2$ we conclude that $\zeta_{\kv} \in H^2 = \DD(H_V)$.

		To complete the first part of the proof, it suffices to show that $(\zeta_{\kv}, \theta_{\kv})$ satisfies the eigenvector equation of $ \HH $ with eigenvalue $ k^2 $: setting $ \theta_{\kv} :=- R_U(k^2-\la)\, W\, \zeta_{\kv}$, one has that  $ \theta_{\kv} \in H^2$ and 
		\beq \label{matita}
			( -\Delta + U +\la -k^2)  \theta_{\kv} =- W\, \zeta_{\kv},
		\eeq
		which combined with \eqref{penna} yields the result.
		
		Conversely, if $ k^2 \in \EE $, then there must be a non-trivial solution to \eqref{geneigen} with $ \varphi_{\kv}, \xi_{\kv} \in H^2 $. Exploiting the assumption on $ \la $, we can invert $ H_U + \la - k^2 $, obtaining
		\beq
			\label{eq:system}
			\begin{cases}
				H_V \varphi_\kv + W \xi_\kv =  k^2 \varphi_\kv, \\    
				\xi_\kv + R_U(k^2 - \la) W \varphi_\kv = k^2 R_U(k^2-\la) \xi_\kv. 
		\end{cases} 
		\eeq
		Now, in order to show that $ \xi_{\kv} $ provides a non-trivial solution to \eqref{eqfonhom}, we would have to invert the first equation, but $ R_V(k^2) $ is defined as a boundary value and therefore it makes sense only on a function belonging to $ L^2_s $, $ s > 1/2 $ (Lemma \ref{agmon3}). More precisely, if $ \lf( H_V - k^2 \ri)\varphi_{\kv} \in L^2_s  $, $ s > 1/2 $, then 
		\bdm
			R_V(k^2) \lf(H_V - k^2\ri) \varphi_{\kv} = \varphi_{\kv},
		\edm
		which together with the first equation of \eqref{eq:system} yields $ \varphi_{\kv} = - R_V(k^2) W \xi_{\kv}$. This in turn can be replaced in the second equation of \eqref{eq:system}, leading to \eqref{eqfonhom}. It just remains to verify that $ \lf( H_V - k^2 \ri)\varphi_{\kv} \in L^2_s  $, $ s > 1/2 $, but this is directly implied by the first equation in \eqref{eq:system}, since $ \xi_{\kv} \in H^2 $ and thus $ W \xi_{\kv} \in L^2_{s} $, for any $ 0 \leq s \leq 3/2 $.
\end{proof}

We are now in position to  prove our  result concerning existence, uniqueness and asymptotic behavior of the  solutions of equation \eqref{geneigen3}.

	\begin{proof}[Proof of Proposition \ref{pro:solution-scattering}]   
		Let us fix $ k^2 \in (0,\la) \setminus \EE$ and $k^2-\la \neq E_j$, for $j=0, \ldots, N-1$. 
		We set  
		\beq
			\eta_{\kv}:=  \varphi_{\kv}- \phi_{V,\kv} 
		\eeq
		and accordingly cast \eqref{geneigen3} as an equation for  $ \eta_{\kv} $  in the space $\BB$, i.e.,
		\beq \label{baba}
			\lf( I - D_{\kv} \ri) \eta_{\kv} = D_{\kv} \phi_{V,\kv}.
		\eeq
		
		As a first step, we shall prove that $D_{\kv} \in \LL^{\infty}(\BB) $. Using the resolvent identity and the notation $T_Z(k)= R_0(k^2) Z$ introduced in Lemma \ref{ike},  we write
		\begin{multline}
 			D_{\kv}  =  R_0(k^2) WR_U (k^2-\la) W - R_0(k^2) V\, R_V(k^2) W  R_U (k^2-\la) W \\ 
= T_W (k)R_U (k^2-\la) W  -T_{\lan x\ran^s V} (k) \,\lan x\ran^{-s} \, R_V(k^2) W R_U (k^2-\la) W,
		\end{multline}
		where we have taken some $ s \in ( 1/2, 3/2]$, which is fixed throughout this proof, and $ \lan x\ran^{-s} $ on the r.h.s. stands for the multiplication operator by the corresponding function. By Lemma \ref{ike} and Assumption \ref{pot}, the operators $T_W (k)$ and $T_{\lan x\ran^s V} (k)$ are compact in $\BB$. Hence it suffices to show that the operators $ R_U(k^2-\la) W $ and $ \lan x\ran^{-s} \, R_V(k^2) W R_U (k^2-\la) W $ are bounded in $ \BB $ to get the result. First we notice that, for any $\eta_{\kv} \in \BB$,
		\begin{equation}
				\lf\| R_U(k^2-\la) W \eta_{\kv} \ri\|_{\BB} \leq  C\, \lf\| R_U(k^2-\la) W \eta_{\kv} \ri\|_{H^2} \leq C \, \lf\| W \eta_{\kv} \ri\|_{2} \leq C\, \lf\|W \ri\|_{2} \lf\|\eta_{\kv} \ri\|_{\BB}
		\end{equation}
		and therefore $R_U(k^2-\la) W \in \mathcal{B}(\BB) $. Analogously, using Lemma \ref{agmon1}, we estimate
		\begin{multline}
			\lf\| \lan x\ran^{-s} \, R_V(k^2) W R_U (k^2-\la) W \eta_{\kv} \ri\|_{\BB} 
\leq C\, \lf\|\lan x\ran^{-s} \, R_V(k^2) W R_U (k^2-\la) W \eta_{\kv} \ri\|_{H^2}  \\
\leq C\, \lf\|R_V(k^2) W R_U (k^2-\la) W \eta_{\kv} \ri\|_{H^2_{-s}} \leq C\, \lf\| W R_U (k^2-\la) W\eta_{\kv} \ri\|_{L^2_s}\\
\leq C \, \lf\|W \ri\|_{L^2_s}  \lf\| R_U (k^2-\la) W\eta_{\kv} \ri\|_{\BB} \leq C \, \lf\| W \ri\|_{L^2_s}   \lf\| W \ri\|_{2} \lf\| \eta_{\kv} \ri\|_{\BB}
		\end{multline}
		and then the operator $\lan x\ran^{-s} \, R_V(k^2) W R_U (k^2-\la) W $ is  also bounded in $ \BB $. 
		Thus, we conclude that $D_{\kv}$ is a compact operator in $\BB$.

		Using the same arguments, one also shows that  $D_{\kv}   \phi_{V,\kv} \in \BB$. Therefore,   by  Freedom alternative and Proposition \ref{poseigen}, we conclude that equation \eqref{baba}  has a unique solution in $\BB$ and this implies that \eqref{geneigen3} has a unique continuous solution $\varphi_{\kv}$, such that $\varphi_{\kv}-\phi_{V,\kv} \in \BB$.
	
		It remains to characterize  the asymptotic behavior of  $\varphi_{\kv} (\xv)$   for $x \rightarrow \infty$. We have shown that $D_{\kv}  \eta_{\kv}$ and  $D_{\kv} \phi_{V,\kv} $ belong to $ \BB $ and therefore both functions are $ o(1) $ as $x\to \infty$. This implies that $ \varphi_{\kv} = \OO(1) $ for large $ x $ and $ W \varphi_{\kv} = \OO(x^{-3}) $, which in particular implies that $ W \varphi_{\kv} \in L^2 $. Moreover, using the resolvent identity, we can rewrite \eqref{geneigen3} as 
		\beq\label{bibo}
			\varphi_{\kv} = \phi_{V,\kv}+ T_{\langle x \rangle W}(k) \langle x \rangle^{-1} R_U (k^2-\la) W\varphi_{\kv}   -T_{\lan x\ran V} (k) \,\lan x\ran^{-1} \, R_V(k^2) W R_U (k^2-\la) W \varphi_{\kv}.
		\eeq
		Now, both functions $ R_U (k^2-\la) W\varphi_{\kv} $ and $ R_V(k^2) W R_U (k^2-\la) W \varphi_{\kv} $ are in $ H^2 $, because $ W \varphi_{\kv} \in L^2 $ and the operators  $ R_U (k^2-\la) $ and $ R_V(k^2) W R_U (k^2-\la) $ map $ L^2 $ to $ H^2 $, and therefore they are bounded in $ \BB $ by Sobolev embedding. We can then apply Lemma \ref{ike}, point c), obtaining
		\bml{
			T_{\langle x \rangle W}(k) \langle x \rangle^{-1} R_U (k^2-\la) W\varphi_{\kv}   -T_{\lan x\ran V} (k) \,\lan x\ran^{-1} \, R_V(k^2) W R_U (k^2-\la) W \varphi_{\kv}  =	\\
			\braketr{\phi_{0,\kvp}}{ (I\!-\!V \, R_V (k^2))  W\,  R_U(k^2-\la)\, W\, \varphi_{\kv}} \frac{e^{i k x}}{4 \pi x} + o(x^{-1}).
		}
		Combining this asymptotic with the expansion \eqref{eq:ls-solution} of $\phi_{V,\kv}$ and recalling that the LS equation can be rewritten as $\phi_{V,\kv}=(I-R_V(k^2) V) \phi_{0,\kv}$, we get the result from \eqref{bibo}.
	\end{proof}

\section{Feshbach resonances}
\label{sec:Feshbach}


According to Proposition \ref{pro:solution-scattering}, the effective LS equation \eqref{geneigen3} admits for $ k > 0 $ a unique solution behaving like a generalized eigenfunction with precise asymptotic for large $ \xv $. In order to derive the behavior of the effective scattering length $ \aeff $, we thus have to investigate the low-energy limit of such quantities. This will be done in the following in two steps: first we will study the LS equation \eqref{geneigen3} at zero energy, i.e., for $ \kv = 0 $, and then we will prove continuity of $ \ampe $, as a function of $ k $, to take the limit $ k \to 0 $.

More precisely, by Proposition \ref{pro:solution-scattering}, we know that for any $ \la \in (0, +\infty) $ and any $ \kv \in \R^3 $, such that $ k^2 \in (0, \la) \setminus \EE $ and $ k^2 - \la \neq E_j $, $ j = 0, \ldots, N-1 $, the LS equation \eqref{geneigen3} admits a unique solution $ \varphi_{\kv} $ in the space of continuous bounded functions and the following asymptotic holds true (see \eqref{eq:varphi-asymptotics} and \eqref{scatamp}):
\bdm
	\varphi_{\kv}(\xv) \underset{x \to +\infty}{\simeq} e^{i \kv \cdot \xv} + \ampe \frac{e^{i k x}}{x},
\edm
with
\beq
	\label{eq:ampe1} 
	\ampe = \tx\frac{1}{4 \pi} \mean{\phi_{V,\kvp}}{W\, R_U(k^2-\la)\, W}{\varphi_{\kv}} + \amp.
\eeq
We prove (Corollary \ref{crit2}) that these results hold true also for $ k = 0 $, i.e., the zero-energy version of \eqref{geneigen3},
\beq \label{zero}
	\varphi_0 - R_V (0) \, W\,  R_U(-\la)\, W\, \varphi_0 = \phi_{V,0},
\eeq
admits a unique continuous solution, whose asymptotic allows to identify the scattering amplitude
\beq
	A_{\mathrm{eff}}(0,0; \la) = \tx\frac{1}{4 \pi} \mean{\phi_{V,0}}{W\, R_U(-\la)\, W}{\varphi_{0}} + a_V.
\eeq
In order to conclude that the expression above coincide with $ \aeff $, we just have to show that $ A_{\mathrm{eff}}(\kv, \kvp; \la) $, and in particular the first term of the r.h.s. of \eqref{eq:ampe1}, is continuous for $ k $ small.

We thus start by considering \eqref{zero} and, more precisely, we focus on the homogeneous equation associated to the one above, which we rewrite as
\beq
	\label{zeroh}
	\eta - R_V (0) \, W\,  R_U(-\la)\, W\, \eta = 0.
\eeq

	\begin{proposition}[Solutions of \eqref{zeroh}]
		\label{crit}
		\mbox{}	\\
		Let Assumptions \ref{pot} and \ref{boundst} hold. Then, there exists at least a critical value $ \la_0 \in (|E_0|, + \infty) $, such that \eqref{zeroh} has a non-trivial solution $ \eta \in \BB $. The number $ M $ of critical values $ \la_j $, for which \eqref{zeroh} admits a non-trivial solution is always finite.
		\newline
		Furthermore, there exists $ \delta_n > 0 $, $ n =1, \ldots N-1 $, with $ \delta_n < \delta_m $, for $ n > m $, such that, if 
		\beq
			\label{eq:Wsmall}
			\lf\| W \ri\|_3 \leq \de_n 
		\eeq
		then there are at least $ n $ critical points $ \la_0, \ldots, \la_n $, for which a non-trivial solution of \eqref{zeroh} occurs, and such values satisfy
		\beq
			\label{eq:order}
			\la_0 > |E_0| > \la_1 > |E_1 |> \cdots > \la_{n} > |E_{n}|.
		\eeq
		Any further critical value $ \la_j $ with $ j > n $ is such that $ \la_n < |E_n| $.
	\end{proposition}
	
	\begin{remark}[Number of critical points]
		\label{rem:Npoints-homo}
		\mbox{}	\\
		The total number $ M $ of critical points $ \la_j $ can not be derived from our hypothesis, but we know by compactness of the operator $ K(\la) $ (see \eqref{unoK}) that $ M $ is finite for all $ \la $. Typically, each $ \mu_j(\la) $ (see the proof below) would provide a unique $ \la_j $ and, in absence of accidental degeneracy, this would result in {\it at least} $ N $ critical points. More precisely, the critical points are identified by the crossing of $ \mu_j(\la) $ with the horizontal line $ 1 $ and therefore the degeneracy of $ \mu_j$'s is relevant only there. This degeneracy is an extremely rare event but its occurrence can not be excluded and this is why in the first part of the statement we refer to at least one critical value. 
		
		Note, however, that even if one assumed $ \lf\| W \ri\|_3 \ll 1 $, there would be no way to conclude that there are {\it exactly} $N$ critical values, since the existence of critical points for $ \la \leq |E_{N-1}| $ depends on the behavior at the origin of the eigenvalues of $ K(\la) $. In fact, if $ H_U $ has a zero energy resonance or eigenstate, we do expect that \eqref{eventual} holds for some $j$ and therefore additional critical points close to the origin can occur.	
	\end{remark}
	
	\begin{remark}[Asymptotics $ \lf\| W \ri\|_3 \ll 1 $]
		\label{rem:asymptW}
		\mbox{}	\\
		When $ \lf\| W \ri\|_3 \ll 1 $, we can perform an asymptotic expansion for any $ \la_j $: we sketch the argument for $\la_0$. According to the behavior of the eigenvalues described in the proof of Proposition \ref{crit}, we expect $\la_0$ to be close to but larger than $ |E_0|$. Let us then restrict the analysis to the region  $ (|E_0| , |E_0|+\de] $, for some $\de > 0 $ fixed. Then, eq. \eqref{uno} becomes
		\beq
			\lf( 1 + \OO(\lf\| W \ri\|_3^2) \ri) \zeta = \frac{1}{E_0 +\la} \braketl{R_V^{1/2} (0) \, W\,\psi_0}{\zeta} R_V^{1/2} (0) \, W\, \psi_0,
		\eeq
		which implies that $ \zeta = R_V^{1/2} (0) \, W\,\psi_0 + \OO(\lf\| W \ri\|_3^3) $ and thus, projecting onto $ R_V^{1/2} (0) \, W\,\psi_0 $, we immediately get
		\bdm
			\frac{1}{E_0 + \la} \lf\| R_V^{1/2} (0) \, W\,\psi_0 \ri\|_2^2 = 1 + \OO(\lf\| W \ri\|_3^2),
		\edm
		which yields
		\be
			\la_0 = |E_0| + \meanlrlr{\psi_0}{W\, R_V (0) \, W}{\psi_0} + {\mathcal O }(\lf\|W \ri\|_3^4).
		\ee
	\end{remark}
	
	\begin{proof}
		The key idea is to rewrite \eqref{zeroh} as a more symmetric equation in $ L^2 $. If $\eta \in \BB$ is a non-trivial solution of \eqref{zeroh}, then $\eta \in H^2_{\mathrm{loc}}$, since $ W \eta \in L^2 $, $ R_U(-\la) $ is bounded in $ L^2 $, the multiplication by $ W $ maps $ L^2 $ into $ L^2_s $, for some $ s > 1 $ (recall Assumption \ref{pot}, point c)) and, finally, $ R_V(0) $ maps $ L^2_s $ to $ H^2_{-s} $ by \eqref{bujo}. In addition, $ (- \Delta+V) \eta \in L^2 $ by \eqref{zeroh}, since $ W R_U(-\la) W \eta \in L^2 $, for any $ \eta \in \BB $. Since $(-\Delta + V)^{1/2}$ is invertible by Assumption \ref{boundst}, we deduce that $ \zeta : =(-\Delta + V)^{1/2} \eta \in L^2 $, and \eqref{zeroh} can be cast in the following form
		\beq 
		\label{uno}
			\zeta - R_V^{1/2} (0) \, W\,  R_U(-\la)\, W\, R_V^{1/2} (0) \, \zeta = 0
		\eeq
		with $R_V^{1/2} (0) \zeta \in \BB $ and $ \zeta \in L^2 $.
		
		To complete the proof of the equivalence of \eqref{zeroh} in $ \BB $ with \eqref{uno} in $ L^2 $, it remains to show that, for any $ \zeta \in L^2 $ solving \eqref{uno}, then $ R_V^{1/2}(0) \zeta \in \BB $ and solves \eqref{zeroh}.
		
		Before proving the other direction of implication, we show that $R_V^{1/2} (0) \, W $ and $W \, R_V^{1/2} (0) $ are compact operators in $L^2$. We recall that, since $ W \in L^3$ by Assumption \ref{pot}, point c), then, by \cite[Thm. XI.22 with $ q  =3$]{RS3}, $ R_0^{1/2} (0) \, W \in \LL^{\infty}(L^2)$ and 
		\beq
			\label{eq:norm-bound1}
			\lf\|R_0^{1/2} (0) \,W \ri\| \leq C \lf\| W \ri\|_3.
		\eeq
		In order to prove that also $ R_V^{1/2} (0) \, W \in \LL^{\infty}(L^2)$, we use the resolvent identity to write
		\beq
			\label{eq:wrw}
			W R_V(0) W  = W R_0 (0) W - W R_0 (0) \, V^{1/2} \lf(I+ |V|^{1/2} R_0 (0) V^{1/2} \ri)^{-1} |V|^{1/2} R_0(0) W
		\eeq
		where we have used the standard notation $ V^{1/2} := \mathrm{sgn}(V) \, |V|^{1/2}$.
By the very same argument used above, $|V|^{1/2} R_0 (0) V^{1/2}$ is compact in $ L^2 $, thanks to Assumption \ref{pot}, point b). If we now combine Assumption \ref{boundst}, point b) with Freedom alternative, we conclude that 
		\bdm	
			\lf(I+ |V|^{1/2} R_0 (0) V^{1/2} \ri)^{-1} \in {\mathcal B}(L^2).
		\edm
		Then, exploiting that $ R_0^{1/2} (0) \, W \in \LL^{\infty}(L^2)$, we deduce from \eqref{eq:wrw} that $ R_V^{1/2} (0) \, W \in \LL^{\infty}(L^2)$ and 
		\beq	
			\label{eq:norm-bound2}		
			\lf\|R_V^{1/2} (0) \,W \ri\| \leq C \lf\| W \ri\|_3.
		\eeq
		
		Now we can complete the argument: for any solution $ \zeta \in L^2 $ of \eqref{uno}, we set $\eta : = R_V^{1/2} (0) \zeta \in L^2 $, so that
		\beq
			\label{eq:eta-eq}	
			\eta = R_V(0) W R_U (-\la ) W R_V^{1/2} (0) \zeta.
		\eeq
		Now, $ R_U (-\la ) W R_V^{1/2} (0) \zeta \in H^2 $ by boundedness of $ W R_V^{1/2} (0) $ in $ L^2 $ and the fact that $  R_U (-\la ) $ maps $  L^2 $ to $ H^2 $. Then, $R_U (-\la ) W R_V^{1/2} (0) \zeta \in \BB$ by Sobolev embedding and $ W R_U (-\la ) W R_V^{1/2} (0) \zeta \in L^2_s $, for any $ s \leq 3 $, so that $ \eta \in H^2_{-s} $, for any $ 1 < s \leq 3 $, thanks to \eqref{bujo}. Hence, $ \eta \in H^2_{\mathrm{loc}} $, but \eqref{eq:eta-eq} also implies that 
		\bdm
			\lf( - \Delta + V \ri) \eta = W R_U (-\la ) W R_V^{1/2} (0) \zeta \in L^2,
		\edm
		and therefore $ \eta \in \dom(H_V) = H^2 \subset \BB $.
		
		So, from now on, we study \eqref{uno} in $L^2$. By compactness of both $ R_V^{1/2} (0) \, W $ and $W \, R_V^{1/2} (0) $ and boundedness of $ R_U(-\la) $, we deduce compactness in $ L^2 $ of the operator
		\bml{
			\label{eq:Kla}
			K(\la) : =   R_V^{1/2} (0) \, W\,  R_U(-\la)\, W\, R_V^{1/2} (0) \\
			= \sum_{j=0}^{N-1} \frac{1}{E_j +\la} \ket{ R_V^{1/2} (0) \, W \, \psi_j } \bra{R_V^{1/2} (0) \,W\, \psi_j   } + \int_0^\infty \: \frac{1}{\mu +\la}  R_V^{1/2} (0) \, W \, \diff E(\mu) \,  W \,  R_V^{1/2} (0),
		}
		where we have used the spectral resolution of $ H_U $ and exploited Assumption \ref{boundst}, point a). Hence, \eqref{uno} can be rewritten as
		\beq
			\label{unoK}
			\zeta = K(\la) \zeta,
		\eeq
		which admits for any $ \la $ at most finitely many non-trivial solutions, because $ K(\la) $ is compact and therefore its spectrum is given by discrete points with finite multiplicity, possibly accumulating only at $ 0 $. Furthermore, since $ K(\la) $ is piecewise continuous and monotone in $ \la $, the number $M $ of critical values $ \la_j $ is always finite.

		To prove the first part of the statement, it is sufficient to notice that $K(\la)$ is a positive compact operator in $L^2$ which is continuous and monotone (decreasing) in $\la\in(|E_0|, +\infty)$. Moreover,
		\beq
			\label{eq:norms}
			\lf\| K(\la) \ri\| \xrightarrow[\la \to | E_0|^+]{} +\infty,		\qquad		\lf\| K(\la) \ri\| \xrightarrow[\la \to +\infty]{} 0.
		\eeq
		Therefore, there must be a non-trivial solution of \eqref{uno} for $ \lambda = \la_0 \in (|E_0|, + \infty) $.

\begin{center}
	\begin{figure}
		\begin{tikzpicture}[scale=1.0]
			\draw[->, very thick] (-1,0) to (8,0) ;
			\draw[->, very thick] (0,-1) to  (0,4) ;
			\node [below] at (8,0) {\scalebox{1.0}{$\la$}};
			\node [left] at (0,1.5) {\scalebox{1.0}{$1$}};
			\node[left] at (0,4) {\scalebox{1.0}{$\mu_j (\la)$} };
			\draw[domain=5.75:8, samples=100,  thick] plot[smooth] (\x, {3/(\x-5) } ); 
			\draw[domain=4:8, samples=100, thick] plot[smooth] (\x, {4/(\x-3) ) } ); 
			\draw[domain=0:8, samples=100,  thick] plot[smooth] (\x, {2.5/(1.9*\x+1) ) } ); 
			\draw[thin, dashed] (5,-1) to (5,3.5);
			\draw[thin, dashed] (3,-1) to (3,3.5);
			\draw (0,1.5) to (8,1.5);
			\node[above right ] at (7,1.5) {$\la_0$ };
			\node[above right] at (5.5,1.5) {$\la_1$ };
			\node[above right] at (0.5,1.5) {$\la_2$ };
			\node[below right] at (5,0) {$|E_0|$ };
			\node[below right] at (3,0) {$|E_1|$ };
		\end{tikzpicture}
		\caption{Typical behavior of the functions $ \mu_j(\la) $, $ j =0, \ldots, N-1 $, and corresponding intersections with the horizontal line $ 1 $.}\label{fig: muj}
	\end{figure}
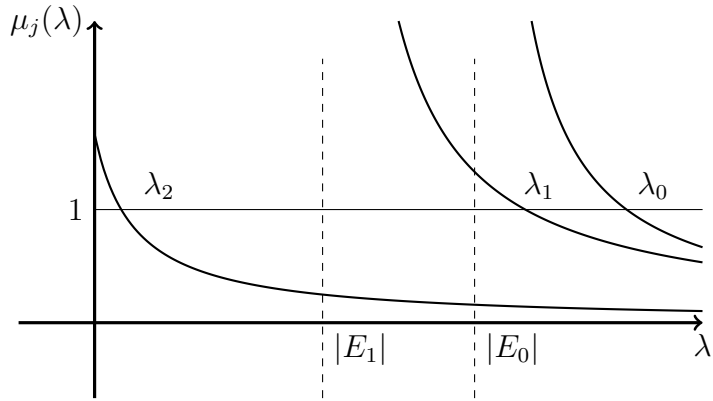
\end{center}

		For $ N > 1 $, we study the behavior of $ K(\la) $ for $ \la < |E_0| $ and show that there exist more critical points $ \la_1, \ldots, \la_{N-1} $, for which a non-trivial solution of \eqref{uno} do exist (see Fig. \ref{fig: muj}). The key point is that, by accidental degeneracy (see Fig. \ref{fig: muj degeneracy}) of the eigenvalue of $ K(\la) $, such points might coincide and the smallness of $ W $ in the condition \eqref{eq:Wsmall} will be used precisely to exclude that this overlap occurs.

\begin{center}
	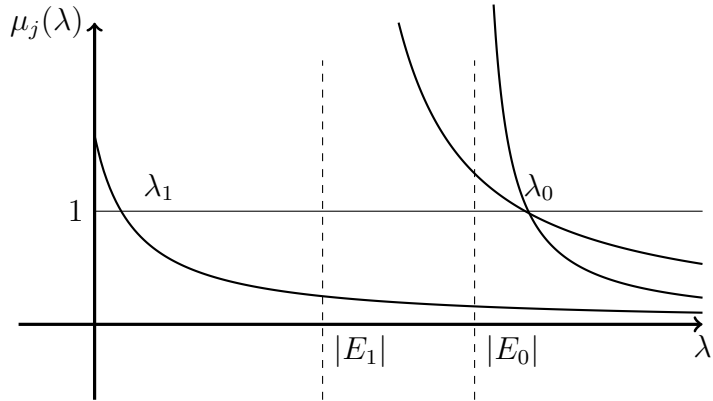
\begin{figure}
		\begin{tikzpicture}[scale=1.0]
			\draw[->, very thick] (-1,0) to (8,0) ;
			\draw[->, very thick] (0,-1) to  (0,4) ;
			\node [below] at (8,0) {\scalebox{1.0}{$\la$}};
			\node [left] at (0,1.5) {\scalebox{1.0}{$1$}};
			\node[left] at (0,4) {\scalebox{1.0}{$\mu_j (\la)$} };
			\draw[domain=5.25:8, samples=100,  thick] plot[smooth] (\x, {1.06/(\x-5) } ); 
			\draw[domain=4:8, samples=100, thick] plot[smooth] (\x, {4/(\x-3) ) } ); 
			\draw[domain=0:8, samples=100,  thick] plot[smooth] (\x, {2.5/(1.9*\x+1) ) } ); 
			\draw[thin, dashed] (5,-1) to (5,3.5);
			\draw[thin, dashed] (3,-1) to (3,3.5);
			\draw (0,1.5) to (8,1.5);
			\node[above right] at (5.5,1.5) {$\la_0$ };	
			\node[above right] at (0.5,1.5) {$\la_1$ };
			\node[below right] at (5,0) {$|E_0|$ };
			\node[below right] at (3,0) {$|E_1|$ };
		\end{tikzpicture}
		\caption{Possible accidental degeneracy at crossing of the $ \mu_j $'s.}\label{fig: muj degeneracy} 
	\end{figure}
\end{center}

		From now assume thus that $ N > 1 $. We are going to show is now that \eqref{uno} admit non-trivial solutions for at least $ N $ values $ \lambda_j \in (|E_j|, + \infty) $. To this purpose we first focus on the behavior near $ |E_0| $ and denote by $\mu_0(\la) $ and $ \xi_0(\la) $, the largest eigenvalue of $K(\la)$ and the corresponding normalized eigenfunction, respectively. Moreover, we set
		\bml{
			\label{eq:Kj}
			K_j (\la) :=  \sum_{i=j}^{N} \frac{1}{E_i +\la} \ket{ R_V^{1/2} (0) \, W \, {\psi}_i } \bra{  R_V^{1/2} (0) \,W\,  {\psi}_i} \\
			+ \int_0^\infty \: \frac{1}{\mu +\la}  R_V^{1/2} (0) \, W \, \diff E(\mu) \,  W \,  R_V^{1/2} (0),
		}
		so that, in particular, 
		\beq
			\label{eq:K}
			K(\la) = \frac{1}{E_0 +\la} \ket{ R_V^{1/2} (0) \, W \, {\psi}_0 } \bra{R_V^{1/2} (0) \,W\,    {\psi}_0} + K_1(\la),
		\eeq
		and the only singular contribution in $ K(\la) $ at $ |E_0| $ is isolated: note that by construction $ K_1 $ remains bounded for $ \lambda = |E_0| $. Obviously, by \eqref{eq:norms},
		\bdm
			\lim_{\lambda \to |E_0|^+} \mu_0(\la) = + \infty,	\qquad 	\lim_{\lambda \to +\infty} \mu_0(\la) = 0,
		\edm
		and thus, if we project \eqref{uno} onto $ \xi_0(\la) $, we get
		\beq
			\label{eq:projected}
			\braket{\xi_0(\la)}{\zeta} = \mu_0(\la) \braket{\xi_0(\la)}{\zeta},
		\eeq
		which is certainly solved by $ \zeta = \xi_0(\la) $ for some $ \lambda_0 \in (|E_0|, + \infty) $, such that $ \mu_0(\la_0) =  1 $. In fact, by taking the derivative w.r.t. $ \la $ of \eqref{eq:Kla}, one can easily realize that such a solution is actually unique, because the operator $ K(\la) $, and thus its eigenvalues, are continuous and monotonically decreasing functions of $ \la \in (|E_0|, + \infty) $. Finally, the min-max characterization of $ \mu_0(\la) $, combined with \eqref{eq:K}, yields
		\beq
			\mu_0(\la) = \sup_{f \in L^2, \lf\| f \ri\|_2 = 1} \meanlrlr{f}{K(\la)}{f} = \lf\| K(\la) \ri\| = \frac{C}{\la - |E_0|} + \OO(1), 	\qquad		\mbox{as } \la \to |E_0|^+,
		\eeq
		so that, if we multiply \eqref{eq:K} by $ E_0 + \la $ and take the limit $ \la \to |E_0|^+ $ using the expansion above, we get that $ \ket{\xi_0(\la)}\bra{\xi_0(\la)} $ converges to the projector onto $ R_V^{1/2} (0) \, W \, {\psi}_0 $, or, equivalently,
		\beq
			\lim_{\la \to |E_0|^+} \lf\| \xi_0(\la) - R_V^{1/2} (0) \, W \, {\psi}_0 \ri\|_2 = 0.
		\eeq
		
		Now, we want to show that one can apply the same argument above to the largest eigenvalue of the operator $ K_1(\la) $ for $ \la \in (|E_1|,+\infty) $, or, more in general, to $ K_j(\la) $ for $ \la \in  (|E_j|,+\infty) $. The idea is to look for a non-trivial solution to \eqref{uno} in 
		\bdm
			\mathscr{H}_1 : = \lf[ \mathrm{span}\lf(R_V^{1/2} (0) \, W \, {\psi}_0 \ri) \ri]^{\perp},
		\edm
		i.e., according to \eqref{eq:K}, we have to consider the homogeneous equation
		\beq
			\label{eq:due}
			\zeta = K_1(\la) \zeta.
		\eeq 	
		Then, if we denote by $ \mu_1(\la) $ the largest eigenvalue of $ K_1(\la) $ and by $ \xi_1(\la) $ the corresponding normalized eigenfunction, we have that
		\beq
			\label{eq:K1}
			K_1(\la) = \frac{1}{E_1 +\la} \ket{ R_V^{1/2} (0) \, W \, {\psi}_1 } \bra{R_V^{1/2} (0)  \,W\,   {\psi}_1} + K_2(\la),
		\eeq
		and $ K_2(\la) $ in bounded for any $ \la \in [|E_1|, + \infty) $. Furthermore,
		\beq
			\lim_{\la \to |E_1|^+} \mu_1(\la) = + \infty,	\qquad		\lim_{\la \to + \infty} \mu_1(\la) = 0,
		\eeq
		where the second identity follows from \eqref{eq:norms}. Hence, as before, $ \xi_1(\la)  $ provides a non-trivial solution of \eqref{eq:due} for some $ \lambda = \lambda_1 \in (|E_1|, + \infty) $ and 
		\beq
			\lim_{\la \to |E_1|^+} \lf\| \xi_1(\la) - R_V^{1/2} (0) \, W \, {\psi}_1 \ri\|_2 = 0.
		\eeq
		Note also that the trivial inequality 
		\bdm
			K(\la) \leq K_1(\la),		\qquad		\mbox{for } \la \in [0, |E_0|],
		\edm
		directly implies that, for $ \la \leq |E_0| $,
		\bmln{
			\sup_{f \in L^2, \lf\| f \ri\|_2 = 1, f \perp  R_V^{1/2} (0) \, W \, {\psi}_0} \meanlrlr{f}{K(\la)}{f} \leq \mu_0(\la) = \sup_{f \in L^2, \lf\| f \ri\|_2 = 1} \meanlrlr{f}{K(\la)}{f} \\
			\leq \sup_{f \in L^2, \lf\| f \ri\|_2 = 1} \meanlrlr{f}{K_1(\la)}{f} = \mu_1(\la) = \sup_{f \in L^2, \lf\| f \ri\|_2 = 1, f \perp  R_V^{1/2} (0) \, W \, {\psi}_0} \meanlrlr{f}{K(\la)}{f},
		}
		and therefore
		\beq
			\mu_1(\la) = \mu_0(\la),		\qquad \mbox{for } \la \in  [0, |E_0|].
		\eeq
		
		The argument can then be easily bootstrapped to show that the equation
		\beq
			\label{eq:tre}
			\zeta = K_j(\la) \zeta,
		\eeq 
		for $ 1 \leq j \leq N-1 $, admits a non-trivial solution in 
		\bdm
			\mathscr{H}_j : = \mathrm{span}\lf(R_V^{1/2} (0) \, W \, {\psi}_0 \,; \ldots; R_V^{1/2} (0) \, W \, {\psi}_j \ri)^{\perp},
		\edm
		for some $ \la_j \in (|E_j|, + \infty) $. Moreover, 
		\beq
			\mu_j(\la) = \mu_0(\la),	\qquad		\mbox{for } \la \in [0, |E_j|].
		\eeq
		Hence, we have found $ \la_0, \ldots, \la_{N-1} $ such that \eqref{uno} has a non-trivial solution for $ \la = \la_j $,  $ j = 1, \ldots, N-1 $. We remark that by monotonicity of $ K_j(\la) $, for $ \la \in (|E_j|, +\infty) $, each $ \la_j $ is unique. Note also that a necessary condition for the existence of further critical points is
		\beq \label{eventual}
			\lim_{\la \to 0^+ } \tilde\mu_{N-1+h} (\la) > 1,	\qquad		\mbox{for some } h > 0,
		\eeq
		where $ \tilde\mu_{N-1+h} $, $ h > 0 $, stands for the lower eigenvalues of the operator $ K_{N-1}(\la) $.

\begin{center}
	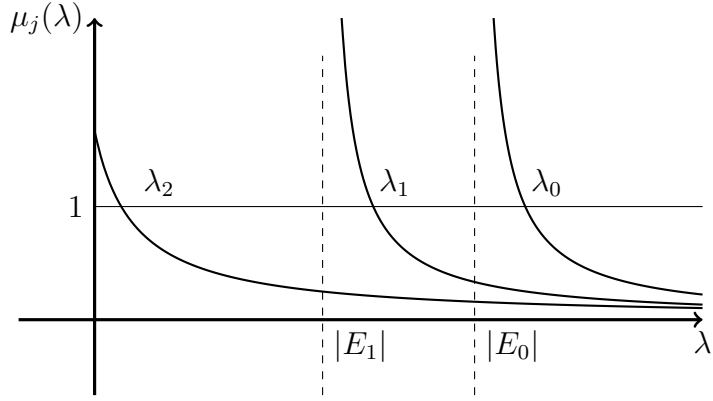
\begin{figure}
		\begin{tikzpicture}[scale=1.0]
			\draw[->, very thick] (-1,0) to (8,0) ;
			\draw[->, very thick] (0,-1) to  (0,4) ;
			\node [below] at (8,0) {\scalebox{1.0}{$\la$}};
			\node [left] at (0,1.5) {\scalebox{1.0}{$1$}};
			\node[left] at (0,4) {\scalebox{1.0}{$\mu_j (\la)$} };
			\draw[domain=5.25:8, samples=100, thick] plot[smooth] (\x, {1/(\x-5) } ); 
			\draw[domain=3.25:8, samples=100,  thick] plot[smooth] (\x, {1/(\x-3) ) } ); 
			\draw[domain=0:8, samples=100, thick] plot[smooth] (\x, {2.5/(1.9*\x+1) ) } ); 
			\draw[thin, dashed] (5,-1) to (5,3.5);
			\draw[thin, dashed] (3,-1) to (3,3.5);
			\draw (0,1.5) to (8,1.5);
			\node[above right ] at (5.6,1.5) {$\la_0$ };
			\node[above right] at (3.6,1.5) {$\la_1$ };
			\node[above right] at (0.5,1.5) {$\la_2$ };
			\node[below right] at (5,0) {$|E_0|$ };
			\node[below right] at (3,0) {$|E_1|$ };
		\end{tikzpicture}
		\caption{Behavior of the functions $ \mu_j(\la) $, $ j =0, \ldots, N-1 $, for $ W $ small.}\label{fig: muj no dege} 
	\end{figure}
\end{center}
		
		As anticipated, however, without any further assumption, we can not exclude a degeneracy of the values $ \mu_j(\la) $: it might indeed be that (Fig. \ref{fig: muj degeneracy})
		\bdm
			\mu_0(\la_j) = \mu_1(\la_j) = \cdots = \mu_j(\la_j),
		\edm
		for some $ j \geq 1 $. As a consequence, a certain number of critical points $ \la_0, \ldots, \la_{j} $ might coincide and thus belong to the interval $ (|E_0|, + \infty) $. By assuming that $ W $ is small enough however we can exclude such a degeneracy (Fig. \ref{fig: muj no dege}): by \eqref{eq:norm-bound1} and \eqref{eq:norm-bound2}, we have (recall \eqref{eq:Kj})
		\beq
			\lf\| K_{n}(|E_j|) \ri\| \leq c_j \lf( \frac{1}{|E_j| + E_n} + 1 \ri) \lf\| W \ri\|_3^2,	\qquad	\mbox{for } n < j \mbox{ and } j = 0, \ldots, n-1,
		\eeq
		which implies that, if
		\beq
			\lf\| W \ri\|_3^2 \leq  c_{n-1}^{-1} \lf( \frac{1}{|E_{n-1}| + E_n} + 1 \ri)^{-1} = : \delta_n,
		\eeq
		then
		\beq
			\lf\| K_{n}(|E_{n-1}|) \ri\| \leq 1, 	
		\eeq
		and therefore $ \la_n \in (E_n, E_{n-1}] $, by monotonicity of $ K_n(\la) $ for $ \la \in (E_n, +\infty) $. This completes the proof. Note that in this second part of the proof we have implicitly used Assumption \ref{boundst}, point c), to deduce that all the vectors $ R_V^{1/2} (0) \, W \, {\psi}_j $, $ j = 0, \ldots, N-1 $, are different, otherwise the argument does not work: if, for instance,  the two vectors did coincide for $ j < k $, then $ \mu_k = \mu_{k+1} $ and there would be one solution less.
	\end{proof}

	We have thus found out that the homogenous equation admits a non-trivial solution only for $ \la = \la_j $ (provided $ \la \neq |E_j| $), $ j = 0, \ldots, M-1 $, for some $ M \in \N $, so that $ M \geq 1 $. The next step is the analysis of the full equation \eqref{zero}: by Freedom alternative we have a unique solution for any $ \la $ different from $ |E_j|, \la_j $, $ j = 0, \ldots, N-1 $.

	\begin{corollary}[Solutions of \eqref{zero}]
		\label{crit2}
		\mbox{}	\\
		Let Assumptions \ref{pot} and \ref{boundst} hold. Then, \eqref{zero} admits a unique continuous solution $\varphi_0 $, which satisfies
		\beq
			\label{eq:varphi-asympt}
			\varphi_0(\xv) = 1 + \frac{\ampeo}{x} + o(x^{-1}),	\qquad		\mbox{as } x \to + \infty,
		\eeq
		for any $ \la \in (0, + \infty) $, with $\la \neq |E_j| $, $ j = 0, \ldots, N-1  $, and $ \la \neq \la_j $, $ j = 0, \ldots, M-1 $, where $ \la_j $ are the critical values as in Proposition \ref{crit}, and where
		\beq \label{scatlen}
			\ampeo = \tx\frac{1}{4\pi} \braketr{\phi_{V,0}}{ W\,  R_U(-\la)\, W \varphi_0} + a_V,
		\eeq
		and $a_V$ the scattering length associated to $ H_V $.
	\end{corollary}

	\begin{proof}
		Taking $ \la $ different from all critical points possibly occurring for the homogeneous equation as in  Proposition \ref{crit} allows to apply the Freedom alternative to the equation
		\beq
			\label{eq:zero1}
			\zeta - K(\la) \zeta = \zeta_{V,0},
		\eeq
		which is the $ L^2$-analogue of \eqref{zero} obtained by setting $ \zeta : = H_V^{1/2} \varphi_0 $ and $ \zeta_{V,0} : = H_V^{1/2} \varphi_{V,0} $: acting as in the first part of the proof of Proposition \ref{crit}, one can indeed show that $ \varphi_0 \in C $ solves \eqref{zero} if and only if $ \zeta \in L^2 $ solves \eqref{eq:zero1}. Then, compactness of $ K(\la) $ proven previously (proof of Proposition \ref{crit}), and Proposition \ref{crit} itself imply that $ I - K(\la) $ is invertible whenever there is no non-trivial solution of the homogenous equation and therefore for $ \la \neq \la_j $, $ j = 0, \ldots, M-1 $. Uniqueness in $ L^2 $ of $ \zeta $ translates into uniqueness in $ C $ of $ \varphi_0 $.
		
		It remains to prove the asymptotic \eqref{eq:varphi-asympt}, but to this purpose it suffices to apply to the operator $ D_0 $ the very same argument used for $ D_{\mathbf{k}} $ in the second part of the proof of Proposition \ref{pro:solution-scattering}, since \eqref{zero} can be rewritten as $ (1 - D_0) \varphi_0 = \phi_{V,0} $. We omit the details for the sake of brevity.
	\end{proof}

In order to complete the derivation of the expression of $ \aeff $, according to Definition \ref{def:scattering-length}, we have to prove that 
\bdm
	\ampeo = \lim_{k \to 0} \ampe,
\edm
i.e., that $ \ampe $ is continuous in $ k $ in a neighborhood of $ k = 0 $. Before proving this fact however we need to show that positive eigenvalues $ \EE $ of $ \HH $ do not accumulate at $ 0 $, under suitable conditions on $ \la $ and $ \kv $.

	\begin{lemma}[$ \EE $ for $ k $ small]
		\label{lemma:positive}
		\mbox{}	\\
		Let Assumptions \ref{pot} and \ref{boundst} hold and let $ \la \in (0, + \infty) $, with $\la \neq |E_j| $, $ j = 0, \ldots, N-1  $, and $ \la \neq \la_j $, $ j = 0, \ldots, M-1 $, where $ \la_j $ are the critical values as in Proposition \ref{crit}. Then, there exists $\de=\de(\la) > 0 $, such that, for any $ 0 < k^2 <\de $, $ \EE \cap [0,\delta) = \emptyset $.
	\end{lemma}
	
	\begin{proof}
		Thanks to Proposition \ref{poseigen}, it suffices to investigate the existence of non-trivial solutions of the homogeneous equation \eqref{eqfonhom}. For $ \kv = 0 $ we have proven in Proposition \ref{crit} that no non-trivial solutions occur for $ \la \neq \la_j $,$ j = 0, \ldots, M $ and $ \la \neq |E_j| $, $ j = 0, \ldots, N-1 $. Therefore, $ I  - D_0 $ is invertible, $ \{ 0 \} \notin \EE $, but $ 0 $ might be an accumulation point of $ \EE $. However, we are going to show that for $ k $ small this is never the case. 
		
		First of all we observe that
		\beq \label{fatto1}
			\lim_{k \to 0} \lf\| R_U (k^2 -\la) - R_U ( -\la) \ri\| =0,
		\eeq
		and thus
		\bdm
			\lim_{\kv \to 0} \lf\| D_{\kv} - D_0 \ri\| = 0.
		\edm
		In fact, by using arguments similar to the ones of the proof of Proposition \ref{pro:solution-scattering}, it is not difficult to see that one also has
		\bdm
			\lim_{\kv \to 0} \lf\| D_{\kv} - D_0 \ri\|_{{\mathcal B}(\BB) }=  0.
		\edm
		Moreover, under the assumptions above, for $ \la $ fixed, there exists $ \delta'(\la) > 0 $, such that for $ k^2 < \delta' $, $R_U (k^2 -\la)$ and $D_{\kv}$ are well defined and 
		\bdm
			I - D_{\kv} = I - D_0 + D_0 - D_{\kv} = (I - D_0) \lf( I + (I - D_0)^{-1} \lf( D_0 - D_{\kv} \ri) \ri).
		\edm
		Now, by the invertibility of $ I - D_0 $ and the estimates above, we can find another $ \delta''(\la) > 0 $, such that
		\bdm
			\lf\| (I - D_0)^{-1} \lf( D_0 - D_{\kv} \ri) \ri\|_{{\mathcal B}(\BB) } < 1,
		\edm
		if $ k^2 \leq \delta''(\la) $, which implies that $ I - D_{\kv} $ is invertible in $ \BB $ by Neumann series. Hence, there is no non-trivial solution of the homogeneous equation \eqref{eqfonhom} in $ \BB $ and thus in $ H^2 $ for $ k^2 < \min(\delta', \delta'') = : \delta $, which proves the result.
	\end{proof}

	\begin{proposition}[Continuity of $ \ampe $]
		\label{usb}
		\mbox{}	\\
		Let Assumptions \ref{pot} and \ref{boundst} hold. Then, for any $ \la \in (0, + \infty) $, with $\la \neq |E_j| $, $ j = 0, \ldots, N-1  $, and $ \la \neq \la_j $, $ j = 0, \ldots, M-1 $, where $ \la_j $ are the critical values as in Proposition \ref{crit}, there exists $\de=\de(\la) > 0 $, such that, for any $ 0 < k^2 <\de $, $ \ampe $ is a continuous function of $ k $ and 
		\beq \label{convergence}
			\ampeo = \lim_{k\to 0 } \ampe =:  \aeff.
		\eeq
	\end{proposition}

	\begin{proof}
		By Lemma \ref{lemma:positive}, we can take $ k $ small enough so that no positive eigenvalue of $ \HH $ lays in the interval $ [0,k^2] $. Then, Proposition \ref{pro:solution-scattering} applies and we know that, if $ \la - k^2 \neq |E_j| $, $ j = 0, \ldots, N-1 $, which again can be ensured by taking $ k $ small enough, since $ \la \neq |E_j| $ by assumption, then \eqref{geneigen3}  admits a unique solution in $ C $, which asymptotically behaves like 
		\bdm
			\varphi_{\kv}(\xv) = e^{i\kv\cdot\xv} + \frac{\ampe}{x} e^{i k x} + o(x^{-1}).
		\edm
		As proven in Corollary \ref{crit2} a similar expansion (see \eqref{eq:varphi-asympt}) holds true for $ \varphi_0 $. Therefore, we just have to control the limit $ k \to 0 $ of the expression above:  we are now going to prove that there exists $ \eps > 0 $ such that
		\beq 
			\label{fatto3}
			\lim_{\kv\to 0} \lf\| \lan x \ran^{-\ve} \lf(  \varphi_\kv - \varphi_0  \ri) \ri\|_\infty = 0,
		\eeq
		which then implies \eqref{convergence} by direct inspection of the asymptotic above and \eqref{eq:varphi-asympt}.
		
		Take then $ \eps > 0 $, so that $\lan x \ran^\ve W \in L^2$. A trivial estimate shows that
		\beq 
			\label{fatto2}
			\lim_{\kv\to 0} \lf\| \lan x \ran^{-\ve} \lf(  \phi_{0,\kv} - \phi_{0,0} \ri) \ri\|_\infty = \lim_{\kv\to 0} \lf\| \lan x \ran^{-\ve} \lf( e^{i \kv \cdot \xv} - 1 \ri) \ri\|_\infty = 0,
		\eeq
		but, by using the LS equation associated to the potential $ V $, i.e., \eqref{eq:ls} and \eqref{eq:ls-solution}, one can also prove that
		\beq 
			\label{fatt3}
			\lim_{\kv\to 0} \lf\| \lan x \ran^{-\ve} \lf(  \phi_{V,\kv} - \phi_{V,0} \ri) \ri\|_\infty = 0,
		\eeq
		Setting $ \varphi_\kv  =  \phi_{V,\kv} + \eta_\kv $, as in the proof of Proposition \ref{pro:solution-scattering}, and consequently $ \varphi_0  =  \phi_{V,0} + \eta_0 $, in order to prove \eqref{fatt3}, it suffices to show that
		\beq \label{fatto4}
			\lim_{\kv\to 0} \lf\| \eta_\kv - \eta_0 \ri\|_\infty = 0.
		\eeq
		by \eqref{fatto2}. On the other hand, we have
		\[
			\eta_{\kv} =  \lf( I - D_{\kv} \ri)^{-1}  D_{\kv} \lan x \ran^{\ve} \; \lan x \ran^{-\ve} \phi_{V, \kv}.
		\]
		We have seen in the proof of Lemma \ref{lemma:positive} that $ \lf\| D_\kv - D_0 \ri\|_{\mathcal{B}(\BB)} \to 0 $, as $ k \to 0 $, but, in fact, if $ \langle x \rangle^{\eps} W \in L^2 $, the result can be strengthened to get
		\bdm
			\lim_{k \to 0} \lf\| \lan x \ran^{\ve} \lf( D_\kv - D_0 \ri) \ri\|_{\mathcal{B}(\BB)} = 0.
		\edm
		Estimate \eqref{fatto4} follows then from invertibility of $ I - D_\kv $ in $ \BB $ for $ k $ small (proof of Lemma \ref{lemma:positive}) and \eqref{fatto2}. 
	\end{proof}

We are now in position to complete the proof of the main Theorem. We recall however first the Schur-Grushin-Feshbach (SGF) formula (see, e.g., \cite{jn}): let  $ X=X_0  \dot{+} X_1$ be a vector space, which is the direct sum of two linear spaces $ X_0 $ and $ X_1 $, then, any linear operator $L$ on $X$ can be expressed as
\beq 
	L= \left(
		\begin{array}{cc}
			L_{00} & L_{01} \\
			L_{10} & L_{11}
		\end{array}
	\right),
\eeq
where $ L_{ij}: X_j \to X_i $, $ i, j = 0,1$.
Suppose now that $L_{00}$ is invertible in $ X_0 $ and set 
\bdm
	S : = L_{11} - L_{10}\, L_{00}^{-1} L_{01},
\edm
which goes under the name of Schur complement of $ L_{11} $. Then, $L $ is invertible if and only if $ S^{-1}$ exists and 
\beq \label{sgf}
	L^{-1}= \left(
		\begin{array}{cc}
			L_{00}^{-1}  + L_{00}^{-1} L_{01} S^{-1} L_{10}L_{00}^{-1}  &- L_{00}^{-1} L_{01}S^{-1} \\
- S^{-1} L_{10}L_{00}^{-1}  &S^{-1}
		\end{array}
\right).
\eeq

	\begin{proof}[Proof of Theorem \ref{main}]
		Fix $\la \in \R^+ $, such that $ \la \neq |E_j|$, $ j = 0, \ldots, N-1 $, and $\la \neq \la_j$, $ j = 0, \ldots, M-1 $. Then, by Propositions \ref{crit} and \ref{usb}, the scattering length $ \aeff $ is well defined and expressed (see \eqref{convergence}) by the formula \eqref{scatlen}, i.e.,
		\beq
			\label{eq:aeff-proof}
			\aeff = \tx\frac{1}{4\pi} \braketr{\phi_{V,0}}{ W\,  R_U(-\la)\, W \varphi_0} + a_V
		\eeq
		which, taking into account \eqref{zero}, can be rewritten as 
		\beq 
			\label{scatlen2}
			\aeff = \tx\frac{1}{4\pi} \braketr{\phi_{V,0}}{ W\,  R_U(-\la)\, W  \left(  I- D_0 \right)^{-1} \phi_{V,0}} +a_V.
		\eeq
		Therefore we have to investigate the inversion of $ I - D_0 $.
		
		The first step towards the result is the analysis of the relation between the operator $ K(\la) $ and the one appearing in \eqref{eq:aeff-proof}, i.e., $ D_0 = R_V(0) W\,  R_U(-\la)\, W $. It is straightforward to show that
		\beq \label{brutta}
  			(I - D_0)^{-1} = \left(  I- R_V (0) \, W\,  R_U(-\la)\, W\, \right)^{-1} = R_V^{1/2} (0) (I-K(\la) )^{-1} H_V^{1/2},
		\eeq
		and, via the identity
		\be
			(I-A)^{-1} = I+A(I-A)^{-1} = I + (I-A)^{-1} A,
		\ee
		we can write 
		\beq \label{bella}
  			(I - D_0)^{-1} = I
  			+  R_V^{1/2} (0)
 \left(  I- K(\la)  \right)^{-1} 
R_V^{1/2} (0)  \, W\,  R_U(-\la)\, W.
		\eeq
		The above formula shows the relation between the inverse of $ I - D_0 $ and the one of $ I - K(\la) $.
	
		Now we analyze the singularities of $ \aeff $, by first addressing the singular points of $ I - K(\la) $. Fix some $ \la_j $, $ j =0, \ldots, M-1 $ and let $ Q_j $ be the orthogonal projector onto
		\bdm
			\mathscr{K}_j : = \ker \lf( I - K(\la_j) \ri) 
		\edm
		while $ P_j : = I - Q_j $ projects onto $ \mathscr{K}_j^{\perp} $.
		Then, we find the following SGF decomposition for the operator $ I-K(\la) $ (recall \eqref{eq:K}):
		\be
			I-K(\la) = 
			\left(
				\begin{array}{cc}
					P_j(I-K(\la))P_j & P_j(I-K(\la))Q_j \\
					Q_j(I-K(\la))P_j & Q_j(I-K(\la))Q_j
				\end{array}
			\right) =:
			\left(
				\begin{array}{cc}
					K_{00}(\la) & K_{01}(\la) \\
					K_{10}(\la) & K_{11} (\la)
				\end{array}
			\right).
		\ee
		By assumption $ K_{00}(\la) $ is invertible at $ \la = \la_j $, and therefore there exists a neighborhood $ (\la_j - \delta, \la_j + \delta) $, $ \delta > 0 $, where it is invertible via Neumann series. Moreover, due to the smoothness of $ K(\la) $, we have, for any such $ \la \in (\la_j-\delta,\la_j+\delta) $,
		\begin{align}
			K_{00}(\la) &= K_{00}(\la_j) +\mathcal{O} (\la-\la_j), \label{00}	\\
 			K_{01}(\la) & = \mathcal{O} (\la-\la_j), \label{01}	\\
			K_{10}(\la) & = \mathcal{O} (\la-\la_j),  \label{10} \\
			K_{11} (\la) & =  (\la-\la_j) Q_j K'(\la_k) Q_j+  \mathcal{O}\lf( (\la-\la_j)^2 \ri), \label{11}	
		\end{align}		
		where the rests have been estimated in operator norm. All the above estimates follow from a Taylor expansion of the operator around $ \la_j $: \eqref{00} is trivial, so let us consider \eqref{01},
		\bmln{
			\lf\| K_{01}(\la) \ri\|^2 = \sup_{f \in L^2, \lf\| f \ri\|_2 = 1} \lf\| K_{01}(\la) f \ri\|_2^2 = \sup_{f \in L^2, \lf\| f \ri\|_2 = 1} \meanlrlr{f}{Q_j (I-K(\la)) P_j (I-K(\la)) Q_j}{f} \\
			\leq \lf\| (I - K(\la)) Q_j \ri\|^2 = \lf\| \lf( K(\la_j) - K(\la) \ri) Q_j \ri\|^2  = \OO\lf((\la - \la_j)^2\ri).
		} 
		The third one \eqref{10} is analogous, so let us consider \eqref{11}:
		\bmln{
			\lf\| K_{11}(\la) - (\la - \la_j) Q_j  K^{\prime}(\la_j) Q_j \ri\|^2 = \lf\| Q_j \lf( K(\la_j) - K(\la) - (\la - \la_j) K^{\prime}(\la_j) \ri) Q_j \ri\|^2 \\
			= \OO\lf( (\la - \la_j)^4 \ri),
		}
		which leads to the last expansion.
		
		\emph{Part (i).} Let us now consider $ \la_0 $, which is defined as the unique point where $ \mu_0(\la_0) = 1 $ (see the proof of Proposition \ref{crit}): at $ \la_0 $ the homogeneous equation $ (I - K(\la_0)) \eta = 0 $ admits a non-trivial solution. A direct computation yields
		\beq
			\label{eq:kerKprime}
			K^{\prime}(\la) = - R_V^{1/2} (0) \, W\,  R^2_U(-\la)\, W\, R_V^{1/2} (0) < 0,
		\eeq
		on $ L^2 $, thanks to the positivity of $ R^2_U(-\la) $ and triviality of the kernel of $ W R_V^{1/2}(0) $, by Assumption \ref{boundst}, point c). Hence, the operator is invertible in $ \mathscr{K}_0 $, which is the key ingredient to apply the SGF method: by \eqref{10}, \eqref{01} and \eqref{11}
		\be
			S = K_{11} (\la) - K_{10}(\la) K_{00}^{-1}(\la)  K_{01}(\la) = (\la-\la_0) \lf[ Q_0 K'(\la_0) Q_0 +  \mathcal{O}\lf(\la-\la_0 \ri) \ri]
		\ee
		which is invertible for $ \la - \la_0  $ small enough, thanks to \eqref{eq:kerKprime}, and
		\be
			S^{-1} = \frac{1}{\la-\la_0} Q_0 Y_0 Q_0 +  \mathcal{O} (1),
		\ee
		where we have denoted by $ Y_0 $ the inverse of $ K'(\la_0) $ restricted to $ \mathscr{K}_0 $. By the SGF formula \eqref{sgf}, we conclude that w.r.t. to the decomposition $ L^2 = \mathscr{K}_0^{\perp} \dot{+} \mathscr{K}_0 $
		\beq
			\lf(I-K(\la) \ri)^{-1} = \frac{1}{\la - \la_0} \lf(
			\begin{array}{cc}
				0	&	0 \\
				0	&	Y_0
		\end{array}
			\ri) + \OO(1),
		\eeq
		where we have estimated
		\begin{align}
			K_{00}^{-1} \lf(I + K_{01} S^{-1} K_{10} K_{00}^{-1} \ri) &= K_{00} \lf( I + \OO\lf( \la - \la_0 \ri) \ri) = \OO(1),	\\
 			K_{00}^{-1} K_{01} S^{-1} & = \mathcal{O} (1), 	\\
			S^{-1} K_{10} K_{00}^{-1} & = \mathcal{O} (1).
		\end{align}	
		Hence, we deduce that w.r.t. to the same decomposition, in a neighborhood of $ \la_0 $,
		\beq
			(I - D_0)^{-1} = \frac{1}{\la - \la_0} \lf(
			\begin{array}{cc}
				0	&	0 \\
				0	&	\tilde{Y}_0
		\end{array}
			\ri) + \OO(1),
		\eeq
		with
		\beq
			\tilde{Y}_0 : = R_V^{1/2} (0)
 Q_0 Y Q_0  R_V^{1/2} (0)  \, W\,  R_U(-\la_0) \, W,
 		\eeq
 		and, consequently
		\beq
			 \aeff = \frac{c_0}{\la - \la_0} + \OO(1).
		\eeq
		The coefficient $ c_0 $ is given by
		\beq
			\label{eq:c0}
			c_0 : = \frac{1}{4\pi} \meanlrlr{R_U(-\la_0) \, W \phi_{V,0}}{   W  R_V^{1/2} (0)
 Q_0 Y_0 Q_0  R_V^{1/2} (0)  \, W\,  }{R_U(-\la_0) \,  W \phi_{V,0}} \neq 0,
		\eeq
		by \eqref{eq: resonant condition} and the invertibility of $ K^{\prime}(\la_0) $. If $ Q_0 = \ket{\eta_0} \bra{\eta_0} $, i.e., $ \mathrm{ran} \, Q_0 $ is one-dimensional, then the above formula simplifies to
		 \beq
 				c_0 =  \frac{1}{4\pi} \mean{\eta_0}{Y_0}{\eta_0} \lf| \braket{\phi_{V,0}}{WR_U(-\la_0) \, W \eta_0} \ri|^2.
		\eeq
		
		\medskip
		
		\emph{Part (ii).} The preliminary assumption we make is that 
			\beq
				\lf\| W \ri\|_3 \leq {\delta_{0} : = \min_{j = 1, \ldots, N-1} \delta_{N-1},}
			\eeq
			where $ \delta_j $ is as in the statement of Proposition \ref{crit}. Hence, by the same Proposition, we know that there exist $ N $ critical values $ \la_j $, $ j = 1, \ldots, N-1 $, satisfying \eqref{eq:order}. Furthermore, no other critical values occur in the interval $ [|E_{N-1}|, +\infty) $, since all the other ones are smaller than $ E_{N-1} $.
			
			Let us now select one critical point $ \la_j $. If we denote by $ Q_j $ the orthogonal projector onto $ \mathscr{K}_j : = \ker(I - K(\la_j)) $ and by $ P_j = I - Q_j $, we can apply the SFG method to invert the operator $ I - K(\la) $ in a neighborhood $ \la \in (\la_j-\delta, \la_j+\delta) $, exactly as we did for $ j = 0 $ in Part (i) of this proof. For instance, the invertibility of $ K^{\prime}(\la_j) $ is again a consequence of \eqref{eq:kerKprime}. In conclusion, we deduce that, for $ |\la - \la_j| < \delta $,
		\beq
			(I - D_0)^{-1} = \frac{1}{\la - \la_j} \lf(
			\begin{array}{cc}
				0	&	0 \\
				0	&	\tilde{Y}_j
		\end{array}
			\ri) + \OO(1),
		\eeq
		with $ \tilde{Y}_j : = R_V^{1/2} (0)
 Q_j Y Q_j  R_V^{1/2} (0)  \, W\,  R_U(-\la_j) \, W $. Therefore,
		\beq
			 \aeff = \frac{c_j}{\la - \la_j} + \OO(1),
		\eeq
		where the coefficient $ c_j \neq 0 $ is explicitly given by
		\beq
			\label{eq:cj}
			c_j : = \frac{1}{4\pi} \meanlrlr{R_U(-\la_j) \, W \phi_{V,0}}{   W  R_V^{1/2} (0)
 Q_j Y_j Q_j  R_V^{1/2} (0)  \, W\,  }{R_U(-\la_j) \,  W \phi_{V,0}},
		\eeq
		which is non-zero due to \eqref{eq: resonant condition perturbative} exactly as above.
		
		To complete the proof, we just have to show that the scattering length is continuous far from the critical values $ \la_j $. This is obvious if $ \la \neq |E_j| $, $ j = 0, \ldots, N-1 $, but it requires some  further discussion when $ \la \to |E_j| $, for some $ j $. We start by observing that given a bounded operator $ L $ on $ L^2 $ such that $ I + L $ is invertible and $ \beta \in \R $, $ \psi, \chi \in L^2 $, then one has
		\beq \label{inv}
			\lf( I + \beta \ket{\chi}\bra{\psi} + L \ri)^{-1}
			= \lf( I + L \ri)^{-1} \lf[ 1 - \frac{1}{\f{1}{\beta} + \braket{\psi}{( I + L )^{-1}\chi}}\ket{\chi} \bra{\lf[( I + L )^{-1}\ri]^* \psi} \ri].
		\eeq
		If we now decompose 
		\[
			R_U (-\la) = \f{1}{E_j+\la} \ket{{{\psi}}_j} \bra{{{\psi}}_j}+ R_j (\la) 
		\]
		with $ R_j(\la) = \OO(1) $, as $ \la \to |E_j| $, then 
		\[
 			I - D_0 = I- R_V (0) \, W\,  R_U(-\la)\, W = I - \f{1}{E_j+\la} \ket{R_V(0) W {\psi}_j} \bra{W {\psi}_j} - \tilde{R}_j (\la),
		\]
		where again $ \tilde{R}_j = \OO(1) $, as $ \la \to |E_j| $. 
		
		Now, we claim that $ I - \tilde{R}_j $ is invertible in a neighborhood of $ E_j $: by compactness of $ \tilde{R}_j $ and Freedom alternative this is equivalent to prove that there is no non-trivial solution to $ (I - \tilde{R}_j) \eta = 0 $. If we set $ \zeta : = H_V^{1/2} \eta $ as in the derivation of \eqref{uno}, we map the above equation into 
		\bdm
			\lf(I - R_V^{1/2} W R_j(\la) W R_V^{1/2} \ri) \zeta = 0.
		\edm	
		The ordering of eigenvalues $ E_i$'s implies that, as an operator, for any $ \la < |E_{j-1}| $,
		\beq
			I -  R_V^{1/2} W R_j(\la) W R_V^{1/2} \geq I - K_{j+1}(\la).
		\eeq
		However, the smallness condition on $ W $ guarantees that $ \lf\| K_{j+1}(|E_j|) \ri\| < 1 $ and thus $ I - K_{j+1}(|E_j|) > 0 $, which in turn implies via the above inequality that $ I -  R_V^{1/2} W R_j(\la) W R_V^{1/2} $ and $ I - \tilde{R}_j $ are both invertible at $ |E_j| $, with bounded inverse.
		
		We can therefore apply \eqref{inv} to get for $ \la $ close to $  |E_j| $
		\bmln{
			\lf( I - D_0 \ri)^{-1} = (I -  \tilde{R}_j (\la) )^{-1} - \frac{1}{E_j + \la + \braketr{W {\psi}_j}{(1 - \tilde{R}_j(|E_j|))^{-1}R_V(0) W {\psi}_j}}	\times \\
			 \times  \ket{(1 - \tilde{R}_j(\la))^{-1} R_V(0) W {\psi}_j} \bra{\lf[(1 - \tilde{R}_j(\la))^{-1}\ri]^* W {\psi}_j},
		}
		Some lengthy computation starting from \eqref{scatlen2} and using the above formula leads to the following expression of the effective scattering length
		\bml{
			\lim_{\la \to |E_j|} \aeff = a_V + \frac{1}{4\pi} \meanlr{W\phi_{V,0}}{ R_j(|E_j|)  }{ W  (I -  \tilde{R}_j (|E_j|) )^{-1} \phi_{V,0}} \\
			 - \frac{\braketr{{\psi}_j}{W (1 - \tilde{R}_j(|E_j|))^{-1} \phi_{V,0}}}{4\pi \braketr{ W{\psi}_j}{(1 - \tilde{R}_j(|E_j|))^{-1} R_V(0) W {\psi}_j}}  \meanlr{ W\phi_{V,0}}{R_j(|E_j|) }{W(1 - \tilde{R}_j(|E_j|))^{-1} R_V(0) W {\psi}_j}.
		}
		Exploiting the decay of $ W $, it is not difficult to see that both $ W \phi_{V,0} $ and $ W (1 - \tilde{R}_j(|E_j|))^{-1} \phi_{V,0}$ belong to $ L^2 $ and therefore all the terms in the expression above are in fact bounded.
	\end{proof}
	
	We finish the Sect. with the proof of Corollary \ref{zerores}.
	
	\begin{proof}[Proof of Corollary \ref{zerores}]
		Fix $ \la = \la_j $ and let $ \varphi\in \BB$ be a non-trivial solution of $ (I - D_0) \varphi =0$ and set
		\beq
			\Psi_0 : = \lf( 
				\begin{array}{c}
					\varphi \\
					- R_U (-\la_j)W  \varphi
				\end{array}
			\ri),
		\eeq
		then by \eqref{zero} and Remark \ref{rem:k02}, 
see also the regularity used in Proposition \ref{crit},
we have $\varphi_0 \in H^2_{-s}$ for $\forall s>1$
and 
\[
(-\Delta + V) \varphi_0 = W R_{U} (-\la) 
W \varphi_0\in L^2.
\]
Moreover $R_U (-\la) W \varphi_0\in H^2$ and
\[
(-\Delta + U+\la) R_U (-\la) W \varphi_0 = W \varphi_0 \in L^2.
\]
Then it is straightforward to verify that $\HH \Psi_0=0$.
	\end{proof}

\bigskip

\nt
{\bf Acknowledgements.} The authors acknowledge the support of INdAM through GNFM.

\bigskip
\bigskip

\end{document}